\documentclass[smallcondensed]{svjour3}
\usepackage [latin1]{inputenc}    % onecolumn (ditto)
\smartqed  % flush right qed marks, e.g. at end of proof
\usepackage{times}
\usepackage{graphicx}
\usepackage{array}
\usepackage[misc]{ifsym}
\usepackage{color}

\usepackage{amsmath,bm}
\usepackage{amssymb}
\usepackage{braket}
\usepackage{bbold}
\usepackage{subfigure}
\usepackage{hyperref}
\usepackage{appendix}
\usepackage{tikz}
\usepackage{ragged2e}
\usepackage{natbib}
\usepgflibrary{shapes.geometric}
\usetikzlibrary{arrows}

\newcommand{\tr}{{\operatorname{Tr}}}
\begin{document}
\title{Unified $(q,s)$-concurrence: a two-parameter family of bipartite entanglement measures
}
\author{Chen-Ming Bai$^{1}$ and Yu Luo$^{2}$}
\institute{Chen-Ming Bai$^{*1}$ and Yu Luo$^{\dag 2}$
\at $^1$Department of Mathematics and Physics, Shijiazhuang Tiedao University, Shijiazhuang 050043, China\\
$^2$School of Artificial Intelligence and Computer Science, Shaanxi Normal University, Xi'an, 710062, China\\
\email{$^*$ baichm@stdu.edu.cn;\quad $^\dag$ penroseluoyu@gmail.com}
}

\date{Received: date / Accepted: date}

\maketitle
\begin{abstract}
Quantum entanglement is a fundamental resource in quantum information processing, yet its quantitative characterization remains challenging, especially for mixed states.
In this work, we introduce a two-parameter family of entanglement measures, referred to as the unified $(q,s)$-concurrence, which generalizes both the standard concurrence and the recently proposed $q$-concurrence. This measure is inspired by the unified entropy and provides a flexible framework for quantifying bipartite entanglement. By combining the positive partial transposition and realignment criteria, we derive an analytical lower bound for this measure for arbitrary bipartite mixed states, revealing a connection to strong separability criteria.
Explicit expressions are obtained for the unified $(q,s)$-concurrence in the cases of isotropic and Werner states under the constraint $q>1$ and $qs\geq 1$. Furthermore, we explore the monogamy properties of the unified $(q,s)$-concurrence for $q\geq 2$, $0\leq s\leq 1$ and $1\leq qs\leq 3$, and establish an entanglement polygon inequality for any multipartite qudit system when $q\geq 1$ and $qs\geq 1$. Our results provide a unified and versatile tool for entanglement quantification, with potential applications in quantum communication and multipartite entanglement theory.
\end{abstract}

\section{Introduction}
Quantum entanglement, a fundamental feature of quantum mechanics, underpins quantum information science and serves as a critical resource for enabling quantum technologies. These include quantum communication \cite{PhysRevLett.134.020802}, quantum cryptography \cite{yin2020entanglement,zeng2023controlled},
quantum teleportation \cite{karlsson1998quantum} and quantum secret sharing \cite{hillery1999quantum,singh2024controlled}. The quantitative characterization of entanglement is a central task within the resource theory of entanglement \cite{vedral1997quantifying}, as it enables the assessment of quantum states for practical applications and deepens our understanding of non-classical correlations. Over the past two decades, numerous entanglement measures have been proposed for bipartite systems, such as concurrence \cite{hill1997entanglement,yang2021parametrized,xuan2025new}, negativity \cite{PhysRevA.58.883,vidal2002computable}, entanglement of formation (EoF) \cite{horodecki2001entanglement,PhysRevA.54.3824}, R{\'e}nyi-$\alpha$ entropy entanglement \cite{san2010monogamy}, Tsallis-$q$ entropy entanglement \cite{kim2010tsallis,luo2016general} and unified $(q,s)$-entropy entanglement \cite{san2011unified}. Each measure captures different aspects of entanglement and is often tailored to specific tasks or state classes. However, the explicit computation of entanglement measures for arbitrary quantum states remains a challenging task, primarily due to the optimization required for mixed states.

So far, analytical results are available only for specific measures--such as in two-qubit systems or certain restricted classes of higher-dimensional states. Several efforts have been made to derive analytical lower bounds for these measures \cite{chen2005entanglement,zhu2012improved,2005Concurrence}. A fundamental question in this context is whether a given bipartite quantum state is entangled or separable. Two widely used operational criteria for detecting entanglement are the positive partial transpose (PPT) criterion \cite{1996Separability} and the realignment criterion \cite{1999Reduction,rudolph2005further,2002Chen}.
For a bipartite state $\rho_{AB}$, separability implies that its partial transpose with respect to subsystem $A$, denoted $\rho^{T_A}$, must be positive semidefinite; conversely, a non-positive partial transpose signals entanglement. The realignment criterion provides another necessary condition: for separable states, the trace norm of the realigned matrix satisfies $\|\mathcal{R}(\rho)\|_1 \leq1$. Both criteria are instrumental in both theoretical and experimental applications of quantum information \cite{horodecki2009quantum}.

In recent years, parameterized entanglement measures have attracted considerable attention due to their enhanced flexibility in characterizing entanglement structures across different quantum states. For example, Yang et al. \cite{yang2021parametrized,yang2022entanglement} introduced the $q$-concurrence for $q\geq 2$, which connects naturally to Tsallis-$q$ entropy and admits an entanglement polygon inequality.
Subsequently, Wei et al.\cite{wei2022estimating} present tight lower bounds of the $q$-concurrence for $q\geq 2$ by exploring the properties of the $q$-concurrence with respect to the positive partial transposition and realignment of density matrices. Motivated by the parameterized $q$-concurrence, Wei and Fei \cite{wei2022parameterized} proposed a new parameterized entanglement measure, the $\alpha$-concurrence with $0\leq \alpha\leq 1/2$. In addition, Xuan et al.\cite{xuan2025new} introduced a novel
parameterized concurrence, the $C^t_q$-concurrence with $q \geq 2$ to solve the limitation that the above entanglement measurements are defined only by the reduced state of one subsystem. Bao et al.\cite{bao2025parameterized}  proposed two types of novel parameterized entanglement measures, in terms of different ranges of parameter, $q$-concurrence with $q > 1$ and $0 < q < 1$. Li et al. \cite{li2024parametrized} proposed parametrized multipartite entanglement measures from the perspective of $k$ nonseparability.

Despite these developments, a unified framework that systematically encompasses diverse parameterized entanglement measures remains elusive. Such a framework is desirable not only for theoretical completeness but also for practical versatility in adapting entanglement quantification to different physical contexts and state families. Notably, the unified $(q,s)$-entropy provides many intriguing applications in
the realms of quantum information theory \cite{san2011unified,li2024monogamy}.
Hence, a natural problem is how to construct an entanglement measure
from general unified $(q,s)$-entropy. Furthermore, the two-parameter structure offers a more flexible characterization of quantum entanglement than its one-parameter counterpart, as it independently controls distinct statistical contributions, enabling higher resolution in distinguishing weakly entangled, strongly entangled, and mixed-state regimes. Physically motivated by generalized entropy frameworks such as Tsallis, R$\acute{e}$nyi, and Sharma-Mittal entropies, this construction unifies several known one-parameter measures as special limiting cases, thereby embedding entanglement quantification into a broader generalized statistical setting. In addition, the additional parameter enhances adaptability in complex quantum systems and quantum resource theories, facilitating a more nuanced description of entanglement under non-Markovian dynamics, long-range correlations, and noisy environments, while improving the detection and discrimination of different entangled state classes.
Motivated by this gap and the expressive power of unified entropy, we introduce in this work the unified $(q,s)$-concurrence, a two-parameter family of bipartite entanglement measures. Our formulation naturally subsumes several existing concurrence variants-including standard concurrence and $q$-concurrence-as special cases, thereby providing a unifying perspective on parameterized entanglement quantification. Beyond generalization, this measure enables new analytical insights, tighter bounds, and extended monogamy relations, offering a flexible and powerful tool for probing entanglement in both idealized and noisy quantum systems.

This paper is organized as follows. In Section \ref{sec:Preliminary}, we propose two types of novel parameterized entanglement measures, in terms of different ranges of parameter, unified $(q,s)$-concurrence with $q\geq1, qs\geq1$ and $0<q<1, 0<qs<1$.
In Section \ref{sec:bound}, we establish rigorous and tight analytical lower bounds for the unified $(q,s)$-concurrence under specific range conditions.  Moreover, we evaluate the unified $(q,s)$-concurrence for isotropic states and Werner states . In Section \ref{sec:Monogamy}, we explore the monogamy properties of unified $(q,s)$-concurrence in multipartite qubit systems for $q\geq1$, $0 \leq s \leq 1$ and $1\leq qs \leq 3$, and we derive an entanglement polygon inequality for the unified $(q,s)$-concurrence with $q\geq1$ and $qs\geq1$. Finally, Section \ref{sec:Conclusion} summarizes the key results and highlights the main findings of our study.
\section{Parametrized entanglement measures}
\label{sec:Preliminary}
Let $\mathcal{H}_A$ and $\mathcal{H}_B$ denote $d$-dimensional Hilbert spaces. For bipartite pure states, a foundational entanglement measure is the concurrence. Given a state $\ket{\psi}_{AB}\in \mathcal{H}_A\otimes \mathcal{H}_B$, its concurrence \cite{hill1997entanglement,rungta2001universal} is defined through the reduced density operator $\rho_A=\tr_B(\ket{\psi}_{AB}\bra{\psi})$ as
\begin{equation}
    \label{eq:concurrence2}
    C(\ket{\psi}_{AB})=\sqrt{2(1-\tr\rho_A^2)}
\end{equation}
quantifying the departure of $\rho_A$ from purity.

Let \(\mathrm{D}(\mathcal{H}_{AB})\) denote the density matrix on \(\mathcal{H}_{AB}\), a well-defined quantum entanglement measure \(\mathcal{E}:\mathrm{D}(\mathcal{H}_{AB})\to\mathbb{R}^+\) should satisfy the following conditions~\cite{horodecki2001entanglement, guhne2009entanglement}:
\begin{enumerate}
    \item[(E1)](\textbf{Non-negativity}) \(\mathcal{E}(\rho) \geq 0\) for any state \(\rho\), where the equality holds only for separable states.
    \item[(E2)](\textbf{LOCC monotonicity}) The measure must not increase, on average, under local operations assisted by classical communication (LOCC). Formally, for any LOCC protocol described by completely positive trace-preserving maps  \(\{M_{i}\}\) with \(\sum_{i} M_{i}^{\dagger} M_{i} = I\), we require
        \begin{equation}
            \mathcal{E}(\rho) \geq \sum_{i} p_{i} \mathcal{E}\big( \rho_{i} \big),
        \end{equation}
        where \(p_{i} = \tr (M_{i} (\rho))\) and \(\rho_{i} = M_{i} (\rho) / p_{i}\).
    \item[(E3)](\textbf{Convexity}) \(\mathcal{E}\) is convex,
        \begin{equation}
            \mathcal{E}\left( \sum_{i} p_{i} \rho_{i} \right) \leq \sum_{i} p_{i} \mathcal{E}\big( \rho_{i} \big)
        \end{equation}
for any probability distribution  \(\{\lambda_{i}\}\) and corresponding states \(\rho_{i}\).
\end{enumerate}

To characterize quantum entanglement, Yang et al.\cite{yang2021parametrized} introduced a parametrized entanglement monotone known as the $q$-concurrence, which is connected to the generalized Tsallis entropy for any $q\geq 2$. For any bipartite pure state $\ket{\psi}_{AB}$ on Hilbert space $\mathcal{H}_A\otimes \mathcal{H}_B$, it is defined as
\begin{equation}
    \label{eq:concurrence21}
    C_q(\ket{\psi}_{AB})=1-\tr\rho^q_A,
\end{equation}
where $\rho_A$ denotes the reduced density matrix of subsystem $A$.

To enrich the characterization of quantum entanglement,
we introduce a family of two-parameter entanglement measures, named as the unified $(q,s)$-concurrence.
\begin{definition}(\textbf{Unified $(q,s)$-concurrence})
\label{def:qsconcur}
    For any bipartite pure state $\ket{\psi}_{AB}$ on Hilbert space $\mathcal{H}_A\otimes \mathcal{H}_B$, the unified $(q,s)$-concurrence is defined as
    \begin{equation}
        \label{eq:concurrenceqs}
        C_{q,s}(\ket{\psi}_{AB})=\epsilon_{q,s}(1-(\tr\rho^q_A)^s),
    \end{equation}
where $\rho_A$ is the reduced density operator of the subsystem $A$ and the sign factor
$\epsilon_{q,s}$ ensures non-negativity:
\begin{equation}
\label{eq:epsilon_qs}
\epsilon_{q,s} =
\begin{cases}
1, & \text{if } q\geq 1\ \text{and } qs\geq 1 , \\
-1, & \text{if } 0<q<1\ \text{ and } 0<qs<1.
\end{cases}
\end{equation}
\end{definition}

\textbf{Remark 1.} The unified $(q,s)$-concurrence encompasses several well-known measures as special cases:
\begin{itemize}
  \item when $s=1$ and $q\geq 2$, i.e., $\epsilon_{q,s}=1$, $C_{q,1}(\ket{\psi}_{AB})=1-\tr\rho^q_A$ reduces to the $q$-concurrence \cite{yang2021parametrized};
  \item when $s=1$ and $0<q\leq 1/2$, i.e., $\epsilon_{q,s}=-1$, $C_{q,1}(\ket{\psi}_{AB})=\tr\rho^q_A-1$ corresponds to the $q$-concurrence for $0<q\leq 1/2$ \cite{wei2022parameterized};
  \item when $s=1$ and $q=2$, we obtain $C_{2,1}(\ket{\psi}_{AB})=1-tr\rho^2_A=\frac{1}{2}C^2(\ket{\psi}_{AB})$, where $C(\ket{\psi}_{AB})$ is the standard concurrence \cite{2005Concurrence}.
\end{itemize}

It is clear that $C_{q,s}(\ket{\psi}_{AB})=0$ if and only if $\ket{\psi}_{AB}$ is a separable state, i.e., $\ket{\psi}_{AB}=\ket{\psi}_A\otimes\ket{\psi}_B$.
The unified $(q,s)$-concurrence is related to the Schatten $q$-norm for positive semidefinite matrices, where the Schatten $q$-norm~\cite{bhatia2013matrix} is defined as
$$\|A\|_q=(\tr A^q)^{1/q}.$$

Any bipartite pure state admits a Schmidt decomposition. Expressing \(|\psi\rangle_{AB}\) in this form yields
\begin{equation}
\label{eq:Schmidt}
|\psi\rangle = \sum_{i=1}^m \sqrt{\lambda_i} |a_i\rangle_A |b_i\rangle_B,
\end{equation}
where $\lambda_i\geq 0$ are the Schmidt coefficients satisfying $\sum_{i=1}^{m}\lambda_i=1$.
Consequently, the reduced density matrices \(\rho_A\) and \(\rho_B\) have the same spectrum \(\{\lambda_i\}\). Therefore,
\begin{equation}
C_{q,s}(\ket{\psi}_{AB})=\epsilon_{q,s}(1-(\tr\rho^q_A)^s)=\epsilon_{q,s}(1-(\tr\rho^q_B)^s).
\end{equation}
Equivalently,
\begin{equation}
\label{eq:epsilonqs}
C_{q,s}(\ket{\psi}_{AB}) = \epsilon_{q,s}(1 - (\sum_{i=1}^m \lambda_i^q)^s),
\end{equation}
where \(0 \leq C_{q,s}(|\psi\rangle_{AB})\leq  \epsilon_{q,s}(1-m^{s(1-q)})\). The lower bound is attained for product states, and the upper bound is saturated by maximally entangled pure states of the form \(\frac{1}{\sqrt{m}} \sum_{i=1}^m |ii\rangle\).

For a mixed state \(\rho_{AB}\) on Hilbert space \(\mathcal{H}_A \otimes \mathcal{H}_B\), we define its unified $(q,s)$-concurrence via the convex-roof extension as follows:
\begin{equation}
\label{eq:concurrencemix}
C_{q,s}(\rho_{AB}) = \inf_{\{p_i,|\psi_i\rangle\}} \sum_i p_i C_{q,s}(|\psi_i\rangle_{AB}),
\end{equation}
where the infimum is taken over all the pure state decompositions of \(\rho_{AB} = \sum_i p_i |\psi_i\rangle_{AB} \langle \psi_i|\), with  \(p_i \geq 0\) and
\(\sum_i p_i = 1\).
\begin{lemma}
\label{le:concave}
The function $F_{q,s}(\rho)=\epsilon_{q,s}(1-(\tr\rho^q)^s)$ is concave,
where \begin{equation}
\epsilon_{q,s} =
\begin{cases}
1, & \text{if } q\geq 1\ \text{and } qs\geq 1 , \\
-1, & \text{if } 0<q<1\ \text{ and } 0<qs<1.
\end{cases}
\end{equation}

\textbf{(1) Concavity.} For any probability distribution $\{p_i\}$ and corresponding density matrices $\rho_i$, the following inequality holds
\begin{equation}
    \label{eq:concave}
    \sum_{i=1}^np_iF_{q,s}(\rho_i)\leq F_{q,s}(\rho),
\end{equation}
where the equality occurs if and only if all $\rho_i$ are identical for all $p_i > 0$.

\textbf{(2) Subadditivity.} For a general bipartite state \(\rho_{AB}\), $q\geq 1$ and $qs\geq 1$, $F_{q,s}(\rho_{AB})=1-(\tr\rho^q_{AB})^s$ satisfies the inequalities:
\begin{equation}
  |F_{q,s}(\rho_{A})-F_{q,s}(\rho_{B})|\leq F_{q,s}(\rho_{AB})\leq F_{q,s}(\rho_{A})+F_{q,s}(\rho_{B}).
  \label{eq:Subadditivity}
\end{equation}
\end{lemma}
\begin{proof}
The proof relies on Minkowski's inequality for positive semidefinite matrices and the convexity/concavity of power functions. A detailed derivation is provided in Appendix  \ref{sec:Appendix}.
\end{proof}
%%%%%%%%%%%%%%%%%%%%%%%%%%%%%%%%%%%%%%%%%%%%%%%%%%%%%%%%%%%
Next we prove the unified $(q,s)$-concurrence defined is a well-defined entanglement measure.
\begin{theorem}
The unified $(q,s)$-concurrence defined in Eq.(\ref{eq:concurrencemix}) constitutes a well-defined parameterized measure of bipartite entanglement.
\end{theorem}
\begin{proof}
\textbf{(E1)} For a pure state \(|\psi\rangle_{AB}\), if it is separable, then its reduced state \(\rho_A\) is pure. This implies that \(\rho_A^q=\rho_A \) and \(\tr\rho_A=\tr\rho_A^q = 1\). Thus,
$$C_{q,s}(|\psi\rangle_{AB})=\epsilon_{q,s}(1-(\tr\rho_A^q)^s)=0.$$

Moreover, any bipartite separable mixed state \(\rho_{AB}\) admits a decomposition into bipartite separable pure states, i.e., \(\rho_{AB} = \sum_i p_i |\psi_i\rangle_{AB}\langle\psi_i|\), where each \( |\psi_i\rangle_{AB} \) is separable. It then follows from Eq.~\eqref{eq:concurrencemix} that \( C_{q,s}(\rho_{AB})=0\).

Conversely, for a pure quantum state \( |\phi\rangle_{AB} \), we have
$$ C_{q,s}(\ket{\phi}_{AB})=\epsilon_{q,s}(1-(\tr\rho^q_A)^s). $$
Let $\{\lambda_i\}_{i=1}^m$ denote the eigenvalues of
\(\rho_A\), satisfying $0\leq \lambda_i \leq 1$ and $\sum_{i=1}^m\lambda_i=1$.
For $q\geq 1$ and $qs\geq 1$, it holds that $(\sum_{i=1}^m\lambda_i^q)^s\leq 1$, whereas for $0<q<1$ and $ 0<qs<1$, we have $(\sum_{i=1}^m\lambda_i^q)^s\geq 1$. Hence, we can obtain that
$$C_{q,s}(\ket{\phi}_{AB})=\epsilon_{q,s}(1-(\tr\rho^q_A)^s)=\epsilon_{q,s}(1-(\sum_{i=1}^m \lambda_i^q)^s)\geq 0.$$
If $C_{q,s}(\ket{\phi}_{AB})=0$, then each \( \lambda_i \) must be either 0 or 1.
Thus, \(\rho_A = |\phi\rangle_A \langle \phi| \) is a pure state, implying that \( |\phi\rangle_{AB} = |\phi\rangle_A \otimes |\phi\rangle_B \) is  separable. Conversely, if \( \rho_{AB} \) is entangled,  then any pure-state decomposition must contain at least one entangled state satisfying \( |\Phi_i\rangle_{AB} \) such that \( C_{q,s}(|\Phi_i\rangle_{AB}) > 0 \), and therefore, \( C_{q,s}(\rho_{AB}) > 0 \).
%%%%%%%%%%%%%%%%%%%%%%%%%%%%%%%%%%%%%%%%%%%%%%%%%%%%%%%%%%%%%%%%%%%%%%%%%%%%

\textbf{(E2)} We first demonstrate that unified $(q,s)$-concurrence $C_{q,s}(\ket{\phi}_{AB})$ is non-increasing on average under LOCC for any pure state $\ket{\phi}_{AB}$. Let $\{M_i\}$ be a completely positive
 trace-preserving map (CPTP) acting on the subsystem $B$. The post-measurement  states are given by
\begin{equation}
    \rho_i^{AB}=\frac{1}{p_i}M_i(\ket{\phi}_{AB}\bra{\phi}),
\end{equation}
with probability $p_i=\tr[M_i(\ket{\phi}_{AB}\bra{\phi})]$. This implies the relation $\rho_A=\sum_ip_i\rho_i^A$, where $\rho_A=\tr_B(\ket{\phi}_{AB}\bra{\phi})$ and $\rho_i^A=\tr_B\rho_i^{AB}$.

Now, let $\{q_{il},\ket{\phi_{il}}_{AB}\}$ denote the optimal pure-state decomposition of $\rho_i^{AB}$, namely, $\rho_i^{AB}=\sum_lq_{il}\ket{\phi_{il}}_{AB}\bra{\phi_{il}}$, satisfying the condition that
\begin{equation}
\label{eq:optimal1}
C_{q,s}(\rho_i^{AB}) = \sum_l q_{il} C_{q,s}(\ket{\phi_{il}}_{AB}),
\end{equation}
where $q_{il}\geq 0$ and $\sum_lq_{il}=1$.

We then derive the following chain of inequalities:
\begin{eqnarray}
\label{eq:purestate}
C_{q,s}(\ket{\phi}_{AB})&=&C_{q,s}(\rho_{A})\nonumber\\
&=&C_{q,s}\Big(\sum_ip_i\rho_i^A\Big)\nonumber \\
&\overset{(a)}{=}&C_{q,s}\Big(\sum_{il}p_iq_{il}\rho_{il}^A\Big)\nonumber\\
&\overset{(b)}{\geq}&\sum_{il}p_iq_{il}C_{q,s}(\rho_{il}^A)\nonumber\\
&\overset{(c)}{=}&\sum_{il}p_iq_{il}C_{q,s}(\ket{\phi_{il}}_{AB})\nonumber\\
&\overset{(d)}{=}&\sum_{i}p_iC_{q,s}(\rho_{i}^{AB}),
\end{eqnarray}
where the validity of (a) follows from $\rho_{il}^A=\tr_B(\ket{\phi_{il}}_{AB}\bra{\phi_{il}})$; (b) relies on the concavity property established in Lemma~\ref{le:concave}; (c) holds follows from Definition \ref{def:qsconcur}; (d) is justified by Eq.(\ref{eq:optimal1}).

For the mixed state $\rho_{AB}$, we consider the post-measurement states $\rho_{i}^{AB}=\frac{1}{p_i}M_i(\rho_{AB})$ with $p_i=\tr[M_i(\rho_{AB})]$. Suppose that $\{t_{j},\ket{v_{j}}_{AB}\}$ represents the optimal pure state decomposition of $\rho_{AB}$, i.e., $\rho_{AB}=\sum_jt_{j}\ket{v_{j}}_{AB}\bra{v_{j}}$. For each $j$, we define $\rho_{ji}^{AB}=\frac{1}{t_{ji}}M_i(\ket{v_j}_{AB})$, which is the final state obtain after the action $M_i$ on $\ket{v_j}_{AB}$ and $t_{ji}=\tr\Big[M_i(\ket{v_j}_{AB})\Big]$.

By Eq.(\ref{eq:purestate}), it follows that
\begin{equation}
\label{eq:purestate1}
    C_{q,s}(\ket{v_{j}}_{AB})\geq \sum_i t_{ji}C_{q,s}(\rho_{ji}^{AB}).
\end{equation}

Due to the linearity of $M_i$, we obtain that
\begin{eqnarray}
  \rho_{i}^{AB}&=&\frac{1}{p_i}M_i(\rho_{AB}) \nonumber\\
  &=&\frac{1}{p_i}M_i\Big(\sum_jt_{j}\ket{v_{j}}_{AB}\Big)\nonumber\\
  &=&\frac{1}{p_i}\sum_jt_{j}M_i(\ket{v_{j}}_{AB})\nonumber\\
  &=&\frac{1}{p_i}\sum_jt_{j}t_{ji}\rho_{ji}^{AB}.
\end{eqnarray}

Now, for each index pair $(j,i)$, let $\rho_{ji}^{AB}=\sum_lt_{jil}\ket{\Phi_{jil}}_{AB}\bra{\Phi_{jil}}$ be the optimal decomposition of $C_{q,s}(\rho_{ji}^{AB})$, i.e.,
\begin{equation}
\label{eq:rhojiAB}
  C_{q,s}(\rho_{ji}^{AB})=\sum_lt_{jil}C_{q,s}(\ket{\Phi_{jil}}_{AB}).
\end{equation}
Hence, we have that
\begin{eqnarray}
    C_{q,s}(\rho_{AB})&=&\sum_jt_jC_{q,s}(\ket{v_j}_{AB})\nonumber\\
    &\overset{(a)}{\geq}& \sum_{ji}t_jt_{ji}C_{q,s}(\rho_{ji}^{AB})\nonumber\\
    &\overset{(b)}{=}& \sum_{jil}t_jt_{ji}t_{jil}C_{q,s}(\ket{\Phi_{jil}}_{AB})\nonumber\\
    &\overset{(c)}{=}& \sum_{i}p_iC_{q,s}(\rho_{i}^{AB}),
\end{eqnarray}
where the validity of (a) follows from Eq.(\ref{eq:purestate1}); (b) holds follows from Eq.(\ref{eq:rhojiAB}); (c) holds follows from $C_{q,s}(\rho_i^{AB})=\sum_{jl}\frac{t_jt_{ji}t_{jil}}{p_i}t_jC_{q,s}(\ket{\Phi_{jil}}_{AB})$ with $\rho_i^{AB}=\frac{1}{p_i}\sum_jt_jt_{ji}t_{jil}\ket{\Phi_{jil}}_{AB}\bra{\Phi_{jil}}$.
%%%%%%%%%%%%%%%%%%%%%%%%%%%%%%%%%%%%%%%%%%%%%%%%%%%%%%%%%%%%%%%%%%%

\textbf{(E3)} Consider a state $\rho_{AB}=\sum_ik_{i}\rho^{AB}_{i}$, where $k_i \geq 0$ are probabilities satisfying $\sum_i k_i = 1$. Let $\rho_i^{AB}=\sum_jk_{ij}\ket{\psi_{ij}}_{AB}\bra{\psi_{ij}}$ be the optimal pure state decomposition for evaluating the concurrence $C_{q,s}(\rho_i^{AB})$. Therefore, we have that
\begin{equation}
\rho_{AB}=\sum_ik_{i}\rho^{AB}_{i}=\sum_{ij}k_{i}k_{ij}\ket{\psi_{ij}}_{AB}\bra{\psi_{ij}}.
\end{equation}

Hence,
\begin{eqnarray}
C_{q,s}\Big(\sum_ik_{i}\rho^{AB}_{i}\Big)&=&C_{q,s}(\rho_{AB})\nonumber\\
&=& \inf_{\{\rho_i,|\psi_i\rangle\}} \sum_i p_i C_{q,s}(|\psi_i\rangle_{AB})\nonumber\\
    &\overset{(a)}{\leq}& \sum_{ij}k_ik_{ij}C_{q,s}(\ket{\psi_{ij}}_{AB})\nonumber\\
    &\overset{}{=}& \sum_{i}k_i\Big(\sum_jk_{ij}C_{q,s}(\ket{\psi_{ij}}_{AB})\Big)\nonumber\\
    &\overset{(b)}{=}& \sum_{i}k_iC_{q,s}(\rho_{i}^{AB}),
\end{eqnarray}
where (a) holds follows from Eq.(\ref{eq:concurrencemix}); (b) holds because $\rho_i^{AB}=\sum_jk_{ij}\ket{\psi_{ij}}_{AB}\bra{\psi_{ij}}$ is the optimal pure state decomposition of $C_{q,s}(\rho_i^{AB})$.
\end{proof}
%%%%%%%%%%%%%%%%%%%%%%%%%%%%%%%%%%%%%%%%%%%%%%%%%%%%%%%%
\section{Establishing lower bounds for the unified $(q,s)$-concurrence}
\label{sec:bound}
Unlike the case of pure entangled states in Eq.(\ref{eq:concurrenceqs}), quantifying entanglement for mixed states remains challenging due to the optimization procedures required \cite{horodecki2009quantum}. Nevertheless, an operationally effective method is presented here to detect the unified $(q,s)$-concurrence for arbitrary bipartite quantum states. In this section, we derive analytical lower bounds for the unified $(q,s)$-concurrence using the PPT criterion and the realignment criterion. Before presenting the lower bound, we briefly review these two separability criteria.

The PPT criterion \cite{1996Separability,Horodecki1996Separability} provides a necessary condition for separability. For a bipartite state $\rho_{AB} = \sum_{ijkl} p_{ij,kl}|ij\rangle \langle kl|$, the partial transposition with respect to subsystem $A$ is defined as
\begin{equation}
\rho^{T_A} = \left( \sum_{ijkl} p_{ij,kl}|ij\rangle \langle kl| \right)^{T_A} = \sum_{ijkl} p_{ij,kl}|kj\rangle \langle il|,
\end{equation}
where $i$ and $k$ denote row and column indices for subsystem $A$, while $j$ and $l$ correspond to subsystem $B$. If $\rho_{AB}$ is separable, then $\rho^{T_A}$ must be positive semidefinite, i.e., $\rho^{T_A} \geq 0$.

The realignment criterion \cite{2002Chen,1999Reduction} offers another separability test. The realignment operation $\mathcal{R}$ acts on the density operator as
\begin{equation}
\mathcal{R}(\rho) = \sum_{ijkl} p_{ij,kl}|ik\rangle \langle jl|.
\end{equation}
For any separable state, the trace norm satisfies $\|\mathcal{R}(\rho)\|_1 \leq 1$, where $\|X\|_1= \tr\sqrt{X^{\dagger}X}$.

These criteria imply entanglement detection when either
\begin{equation}
\label{eq:trace1}
\|\rho^{T_A}\|_1 > 1 \quad \text{or} \quad \|\mathcal{R}(\rho)\|_1 > 1.
\end{equation}

For the pure state $\ket{\psi}$ in Eq.(\ref{eq:Schmidt}), it can be shown \cite{2005Concurrence} that both norms reduce to
\begin{equation}
\|\rho^{T_A}\|_1 = \|\mathcal{R}(\rho)\|_1 = \left( \sum_{i=1}^m \sqrt{\lambda_i} \right)^2\leq m,
\label{eq:norms1}
\end{equation}
where $\{\lambda_i\}$ forms a probability distribution.
%%%%%%%%%%%%%%%%%%%%%%%%%%%%%%%%%%%%%%%%%%%%%%%%%%%%%%%%%%%%%
\begin{theorem}
\label{th:bound1}
For any mixed entanglement state $\rho$ on Hilbert
space $H_A\otimes H_B$, the unified $(q,s)$-concurrence defined in Eq.(\ref{eq:concurrencemix}) satisfies the following inequality:
\begin{equation}
\label{eq:the41}
C_{q,s}(\rho) \geq \frac{1-m^{s(1-q)}}{1-m^{-s}}\left[1-\left(1-\frac{\left( \max\left\{ \|\rho^{T_A}\|_1, \|\mathcal{R}(\rho)\|_1 \right\} - 1 \right)^2}{m(m-1)}\right)^s\right]
\end{equation}
for either $m\geq 2$, $q\geq 2$, $s\geq1.1391$ or $m\geq 2$, $s\geq 1$, $q\geq 2.4721$.
\end{theorem}

\begin{proof}
    For a pure state $\ket{\psi}$ given in Eq.(\ref{eq:Schmidt}), let us analyze the monotonicity of the function with respect to $q$ under the fixed $s$,
    \begin{equation}
        f(q,s)=\frac{1-(\sum_{i=1}^m\lambda_i^q)^s}{1-m^{s(1-q)}}
    \end{equation}
for any $q\geq 2$ and $qs\geq 1$. The first derivative of
$f(q,s)$ with respect to $q$ is given by
\begin{equation}
    \frac{\partial f}{\partial q}=\frac{G(m,q,s)}{(1-m^{s(1-q)})^2},
\end{equation}
where
\begin{eqnarray*}
   G(m,q,s)=s(m^{s(1-q)}-1)\Big(\sum_{i=1}^m\lambda_i^q\Big)^{s-1}\Big(\sum_{i=1}^m\lambda_i^q\ln\lambda_i\Big)-sm^{s(1-q)}\ln m\Big(1-\Big(\sum_{i=1}^m\lambda_i^q\Big)^s\Big).
\end{eqnarray*}
Employing the Lagrange multiplies \cite{2003Concurrence} under constraints $\sum_{i=1}^m\lambda_i=1$ and $\lambda_i>0$, one has that there is only one stable point $\lambda_i=\frac{1}{m}$ for every $i = 1,\cdots, m,$
for which $G(m,q,s)=0$ for any $q\geq 2$ and $qs\geq 1$. Since the second
derivative at this point is given by
\begin{eqnarray*}
    \frac{\partial^2G(m,q,s)}{\partial\lambda_i^2}\Big|_{\lambda_i=\frac{1}{m}}
    = sm^{2+s(1-q)}\Big((2qs-1)(m^{s(1-q)}-1)+sq(qs-1)\ln m\Big).
\end{eqnarray*}
Therefore, we get that
\begin{equation}
\label{eq:inequality30}
    \frac{\partial^2G(m,q,s)}{\partial\lambda_i^2}\Big|_{\lambda_i=\frac{1}{m}}\geq0.
\end{equation}

The inequality (\ref{eq:inequality30}) translates into a condition on $q$ and $s$ for a given
$m$. Solving this inequality for $m=2$ (the qubit case, which often gives the tightest constraint) yields the numerical thresholds $s\geq 1.1391$ (for fixed $q=2$) or $q\geq 2.4721$
(for fixed $s=1$). For higher $m$, the constraints are looser. Hence, the inequality (\ref{eq:inequality30}) holds
for either $m\geq 2$, $q\geq 2$, $s\geq1.1391$ or $m\geq 2$, $s\geq 1$, $q\geq 2.4721$. In these cases, the minimum extreme point is the minimum point and $\frac{\partial f}{\partial q} > 0$. Therefore, $f(q,s)$ is  an increasing function with respect to $q$ for the fixed $s$.
Hence, we can obtain that
\begin{eqnarray}
\label{eq:purebound}
    C_{q,s}(\ket{\psi})&\geq&\frac{1-m^{s(1-q)}}{1-m^{-s}}C_{2,s}(\ket{\psi})\nonumber\\
    &=&\frac{1-m^{s(1-q)}}{1-m^{-s}}\left(1-\Big(\sum_{i=1}^m\lambda_i^2\Big)^s\right)\nonumber\\
    &=&\frac{1-m^{s(1-q)}}{1-m^{-s}}\left(1-\Big(1-2\sum_{i<j}\lambda_i\lambda_j\Big)^s\right)\\
    &\geq&\frac{1-m^{s(1-q)}}{1-m^{-s}}\left(1-\Big(1-\frac{(\|\rho^{T_A}\|_1-1)^2}{m(m-1)}\Big)^s\right)\nonumber
\end{eqnarray}
where $\rho=\ket{\psi}\bra{\psi}$, and the last inequality is due to that $\sum_{i<j}\lambda_i\lambda_j\geq \frac{(\|\rho^{T_A}\|_1-1)^2}{2m(m-1)}$ \cite{2005Concurrence}.

Assume $\rho=\sum_ip_i\ket{\varphi_i}\bra{\varphi_i} $ is
the optimal pure state decomposition for $C_{q,s}(\rho)$. For either $m\geq 2$, $q\geq 2$, $s\geq1.1391$ or $m\geq 2$, $s\geq 1$, $q\geq 2.4721$, we have that
  \begin{eqnarray}
  \label{eq:oppbound1}
    C_{q,s}(\rho)&=&\sum_ip_iC_{q,s}(\ket{\varphi_i})\nonumber\\
    &\geq&\frac{1-m^{s(1-q)}}{1-m^{-s}}\sum_ip_i\left(1-\Big(1-\frac{(\|\rho_i^{T_A}\|_1-1)^2}{m(m-1)}\Big)^s\right)\nonumber\\
    &=&\frac{1-m^{s(1-q)}}{1-m^{-s}}\left(1-\sum_ip_i\Big(1-\frac{(\|\rho_i^{T_A}\|_1-1)^2}{m(m-1)}\Big)^s\right)\nonumber\\
    &\geq&\frac{1-m^{s(1-q)}}{1-m^{-s}}\left(1-\Big(1-\frac{\sum_ip_i(\|\rho_i^{T_A}\|_1-1)^2}{m(m-1)}\Big)^s\right)\nonumber\\
    &\geq&\frac{1-m^{s(1-q)}}{1-m^{-s}}\left(1-\Big(1-\frac{(\sum_ip_i\|\rho_i^{T_A}\|_1-1)^2}{m(m-1)}\Big)^s\right)\nonumber\\
    &\geq&\frac{1-m^{s(1-q)}}{1-m^{-s}}\left(1-\Big(1-\frac{(\|\rho^{T_A}\|_1-1)^2}{m(m-1)}\Big)^s\right),
\end{eqnarray}
where $\rho_i=\ket{\varphi_i}\bra{\varphi_i}$. The first inequality is from Eq.(\ref{eq:purebound}); the second and third inequalities are obtained from the convexity of the function $f(x)=x^s$ with $s\geq1$; the last inequality is due to the
convex property of the trace norm and $\|\rho^{T_A}\|_1\geq 1$ in Eq.(\ref{eq:trace1}).

Similar to Eq.(\ref{eq:purebound}) and Eq.(\ref{eq:oppbound1}), we obtain from Eq.(\ref{eq:norms1}) that
\begin{equation}
    \label{eq:oppbound2}
     C_{q,s}(\rho)\geq\frac{1-m^{s(1-q)}}{1-m^{-s}}\left(1-\Big(1-\frac{(\|\mathcal{R}(\rho)\|_1-1)^2}{m(m-1)}\Big)^s\right)
\end{equation}
for either $m\geq 2$, $q\geq 2$, $s\geq1.1391$ or $m\geq 2$, $s\geq 1$, $q\geq 2.4721$.

Combining Eq.(\ref{eq:oppbound1}) and Eq.(\ref{eq:oppbound2}), we complete the proof.
\end{proof}
%%%%%%%%%%%%%%%%%%%%%%%%%%%%%%%%%%%%%%%%%%%
\begin{theorem}
\label{th:bound2}
For any mixed entanglement state $\rho$ on Hilbert
space $H_A\otimes H_B$, the unified $(q,s)$-concurrence defined in Eq.(\ref{eq:concurrencemix}) satisfies the following inequality:
 \begin{equation}
\label{eq:the42}
C_{q,s}(\rho) \geq \frac{m^{s(1-q)}-1}{m^{s}-1}\left[(\max\left\{ \|\rho^{T_A}\|_1, \|\mathcal{R}(\rho)\|_1 \right\})^s-1\right],
\end{equation}
where $m\geq 2$, $0<q< 1$ and $0<qs<0.9066$.
\end{theorem}
\begin{proof}
    The proof of the first part of Theorem \ref{th:bound2} is similar to that of Theorem \ref{th:bound1}. We can obtain that the function
    \begin{equation*}
        f(q,s)=\frac{(\sum_{i=1}^m\lambda_i^q)^s-1}{m^{s(1-q)}-1}
    \end{equation*}
is a decreasing function with respect to $q$ for $m\geq 2$, $0<q< 1$ and $0< qs<0.9066$.
Hence, we have that
\begin{eqnarray}
\label{eq:purebound2}
    C_{q,s}(\ket{\psi})&\geq&\frac{m^{s(1-q)}-1}{m^{s/2}-1}C_{\frac{1}{2},s}(\ket{\psi})\nonumber\\
    &=&\frac{m^{s(1-q)}-1}{m^{s/2}-1}\left(\Big(\sum_{i=1}^m\sqrt{\lambda_i}\Big)^s-1\right)\nonumber\\
    &\geq&\frac{m^{s(1-q)}-1}{m^{s/2}-1}\cdot\frac{(\sum_{i=1}^m\sqrt{\lambda_i})^{2s}-1}{m^{s/2}+1}\nonumber\\
    &=&\frac{m^{s(1-q)}-1}{m^{s}-1}\Big(\|\rho^{T_A}\|_1^s-1\Big),
\end{eqnarray}
where $\rho=\ket{\psi}\bra{\psi}$, and the last inequality is due to that $1\leq\sum_{i=1}^m\sqrt{\lambda_i}\leq \sqrt{m}$. Since $s\geq 1$, then $(\sum_{i=1}^m\sqrt{\lambda_i})^s\leq (\sqrt{m})^s$.

Let $\rho=\sum_ip_i\ket{\varphi_i}\bra{\varphi_i} $ be
the optimal pure state decomposition for $C_{q,s}(\rho)$. For $m\geq 2$, $0<q< 1$ and $0 < qs<0.9066$, we have that
  \begin{eqnarray}
  \label{eq:oppbound3}
    C_{q,s}(\rho)&=&\sum_ip_iC_{q,s}(\ket{\varphi_i})\nonumber\\
    &\geq&\frac{m^{s(1-q)}-1}{m^{s}-1}\sum_ip_i\Big(\|\rho_i^{T_A}\|_1^s-1\Big)\nonumber\\
    &\geq&\frac{m^{s(1-q)}-1}{m^{s}-1}\Big(\big(\sum_ip_i\|\rho_i^{T_A}\|_1\big)^s-1\Big)\nonumber\\
    &\geq&\frac{m^{s(1-q)}-1}{m^{s}-1}\Big(\|\rho^{T_A}\|_1^s-1\Big),
\end{eqnarray}
where $\rho_i=\ket{\varphi_i}\bra{\varphi_i}$. The first inequality is from Eq.(\ref{eq:purebound2}); the second inequality is obtained from the convexity of the function $f(x)=x^s$ with $s\geq1$; the last inequality is due to the
convex property of the trace norm and $\|\rho^{T_A}\|_1\geq 1$ in Eq.(\ref{eq:trace1}).

Similar to Eq.(\ref{eq:purebound2}) and Eq.(\ref{eq:oppbound3}), we obtain from Eq.(\ref{eq:norms1}) that
\begin{equation}
    \label{eq:oppbound4}
     C_{q,s}(\rho)\geq \frac{m^{s(1-q)}-1}{m^{s}-1}\Big(\|\mathcal{R}(\rho)\|_1^s-1\Big)
\end{equation}
for $m\geq 2$, $0<q< 1$ and $0 < qs<0.9066$.
Combining Eq.(\ref{eq:oppbound3}) and Eq.(\ref{eq:oppbound4}), we complete the proof.
\end{proof}
\textbf{Remark 2.} When $s=1$, the bound of Theorem \ref{th:bound1} is consistent with the bound in Ref.\cite{wei2022estimating}. On the other hand, when $s=1$, the bound of Theorem \ref{th:bound2} for either $0<q\leq \frac{1}{2}$ is consistent with the bound in Ref.\cite{wei2022parameterized}.
%%%%%%%%%%%%%%%%%%%%%%%%%%%%%%%%%%%%%%%%%%%%%%%%%%%%%%%%%%%
\begin{example}
A paradigmatic class of states useful for testing entanglement measures are the isotropic states\cite{1999Reduction}. Defined on the Hilbert space $\mathbb{C}^d \otimes \mathbb{C}^d$
, these states exhibit symmetry under simultaneous local unitary transformations of the form
$U \otimes U^*$ for any unitary operator $U$. Parameterized by a fidelity parameter $F\in[0,1]$
 with respect to the maximally entangled state $|\Psi^+\rangle =\frac{1}{\sqrt{d}} \sum_{i=1}^d |ii\rangle$,  an isotropic state is expressed as
\begin{equation}
\rho_F = \frac{1 - F}{d^2-1} \left( I - |\Psi^+ \rangle \langle \Psi^+| \right) + F |\Psi^+ \rangle \langle \Psi^+|,
\label{eq:isotropic_state}
\end{equation}
where $I$ is the identity operator on $\mathbb{C}^d \otimes \mathbb{C}^d$.
The fidelity $F$ of the state $\rho_F$ with respect to $\rho_\Psi =|\Psi^+\rangle\langle\Psi^+|$ is defined
by
$$F=f_{\Psi^+}(\rho_F)=\langle \Psi^+ |\rho_F |\Psi^+ \rangle,$$
where $F$ lies in the range $0\leq F\leq 1$. The
isotropic state $\rho_F$ is separable when $F\leq \frac{1}{d}$.

For isotropic states $\rho_F$, several analytic results are known for entanglement measures, including the entanglement of formation~\cite{2000Entanglement}, the tangle and concurrence \cite{2003Concurrence}, and the
R{\'e}nyi $\alpha$-entropy entanglement \cite{wang2016entanglement}. Moreover, it has been shown that when $F \leq \frac{1}{d}$, the trace norms of the partial transpose $\rho^{T_A}$ and the realigned matrix $\mathcal{R}(\rho)$ coincide, satisfying
$\|\rho^{T_A}\|_1=\|\mathcal{R}(\rho)\|_1=dF$ \cite{rudolph2005further,vidal2002computable}.

Inspired
by the techniques \cite{yang2021parametrized,2000Entanglement,2003Concurrence,wang2016entanglement}, the unified $(q,s)$-concurrence $C_{q,s}(\rho_F)$ for
these states will be derived by an extremization as follows.

\begin{lemma}
\label{lem:isotropic states}
The unified $(q,s)$-concurrence for isotropic states $\rho_F$ defined on the Hilbert space $\mathbb{C}^d \otimes \mathbb{C}^d$ (with $d \geq 2$) is given by
\begin{equation}
C_{q,s}(\rho_F) = \mathrm{co}[\xi(F, q,s, d)],
\label{eq:qs_concurrence}
\end{equation}
where $F\in(\frac{1}{d}, 1]$ and $\mathrm{co}(\cdot)$ denotes the convex envelope (the largest convex function that is upper bounded by the given function). The function $\xi(F, q,s, d)$ is defined as
\begin{equation}
\xi(F, q,s, d) = 1 - (\gamma^{2q} + (d - 1)\delta^{2q})^s,
\label{eq:xi_function}
\end{equation}
where $q\geq 1$, $qs\geq1$, $\gamma$ and $\delta$ satisfying
\begin{equation*}
 \gamma = \frac{1}{\sqrt{d}}\left(\sqrt{F}+\sqrt{(d-1)(1-F)}\right), \ \
\delta = \frac{1}{\sqrt{d}}\left(\sqrt{F}-\frac{\sqrt{1-F}}{\sqrt{d-1}}\right).
\end{equation*}
\end{lemma}

The derivation of Eq. (40) leverages the symmetry of isotropic states and involves a constrained minimization over Schmidt coefficients, detailed in Appendix \ref{sec:Appendix A}. We now demonstrate the tightness of the lower bound in Eq.(\ref{eq:the41}) by applying it to isotropic states.

(i) For the special case $q=2$, $s=2$ and $d=2$, the expression in Eq.(\ref{eq:xi_function}) reduces to
 \begin{equation}
 \xi(F,2,2,2) =1-\left(\frac{1}{2}+2F(F-1)\right)^2,
\label{eq:xi_function21}
 \end{equation}
which holds for \( F\in (\frac{1}{2}, 1]\).

For \( F\geq 0.908\), we find \(\frac{d^2}{d F^2} \xi(F,2,2,2) < 0\), indicating that \(\xi(F, 2,2,2)\) is no longer convex on the interval \(F \in (0.908, 1]\). As \(C'_{2,2}(\rho_F)\) is defined as the largest convex function that is bounded above by Eq.(\ref{eq:qs_concurrence}), we ensure convexity by linearly interpolating between the points at \(F = 0.908\) and \(F = 1\). Therefore,
\begin{equation*}
    C'_{2,2}(\rho_F) =
\begin{cases}
0, & F \leq 0.5, \\
1-\left(\frac{1}{2}+2F(F-1)\right)^2, & 0.5 < F \leq 0.908, \\
2.119F-1.369, & 0.908 < F \leq 1.
\end{cases}
\end{equation*}
This result agrees exactly with the lower bound given in Eq.(\ref{eq:the41}) for $0.5<F\leq 0.908$.

(ii) For the special case $q=2$, $s=2$ and $d=3$, Eq.(\ref{eq:xi_function}) becomes
\begin{equation}
\xi(F, 2,2, 3) = 1 - (\gamma^{4} + 2\delta^{4})^2,
\label{eq:xi_function223}
\end{equation}
with the parameters $\gamma$ and $\delta$ satisfying
\begin{eqnarray*}
 \gamma &=& \frac{1}{\sqrt{3}}\left(\sqrt{F}+\sqrt{2(1-F)}\right), \\
\delta &=& \frac{1}{\sqrt{3}}\left(\sqrt{F}-\sqrt{\frac{1-F}{2}}\right),
\end{eqnarray*}
where $F\in(\frac{1}{3},1]$.
\begin{figure*}[ht]
\centering
    \includegraphics[width=0.65\linewidth]{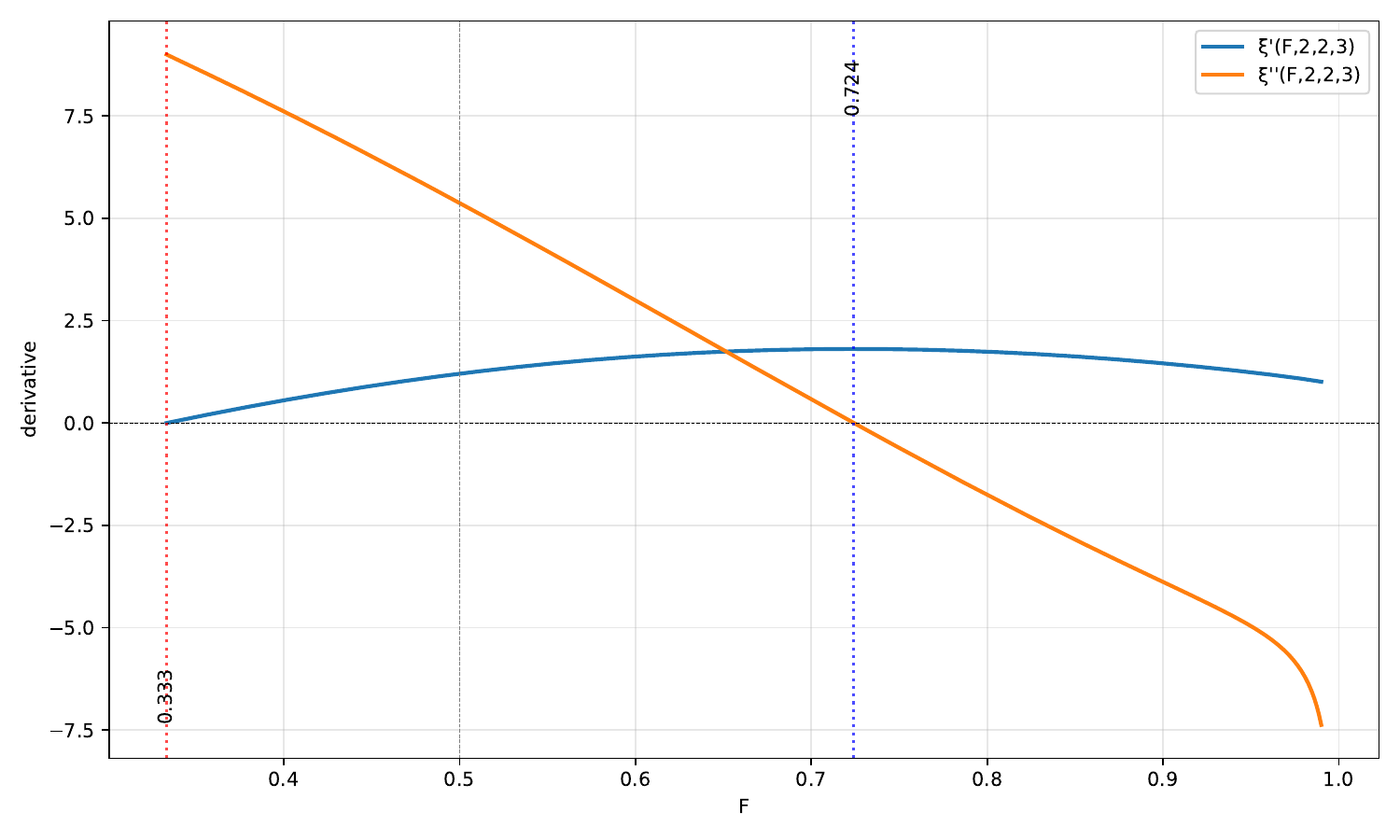}
     \caption{First and second derivatives of $\xi(F, 2,2, 3)$ with respect to $F$.}
    \label{fig:figure223}
\end{figure*}

Since \(\frac{d}{d F} \xi(F, 2,2, 3)>0\), the function \(\xi(F, 2,2, 3)\) is monotonically increasing in the region where \(\rho_F\) is entangled, as illustrated in Figure \ref{fig:figure223}. For \(F \geq 0.724\), we find \(\frac{d^2}{d F^2} \xi(F, 2,2, 3) < 0\), indicating that \(\xi(F, 2,2, 3)\) is no longer convex on the interval \(F \in (0.724, 1]\). As \(C_{2,2}(\rho_F)\) is defined as the largest convex function that is bounded above by Eq.(\ref{eq:qs_concurrence}), we ensure convexity by linearly interpolating between
the points at \(F = 0.724\) and \(F = 1\). Therefore,
\begin{equation}
\label{eq:c22function}
    C_{2,2}(\rho_F) =
\begin{cases}
0, & F \leq 1/3, \\
\xi(F, 2,2, 3), & 1/3 < F \leq 0.724, \\
1.52F-0.63, & 0.724 < F \leq 1.
\end{cases}
\end{equation}

From Theorem \ref{th:bound1}, we derive the following lower bound
\begin{equation}
\label{eq:c22bound}
    C_{2,2}(\rho_F)\geq 1-\left(\frac{5}{6}-\frac{3F^2-2F}{2}\right)^2.
\end{equation}

The $q$-concurrence of the isotropic state was derived in Ref.\cite{yang2021parametrized},
where it is given by
\begin{equation}
\label{eq:c2function}
C_2 (\rho_F) =
\begin{cases}
0, & F \leq 1/3, \\
1-\gamma^4-2\delta^4, &1/3< F \leq 8/9\\
\frac{3F}{2}-\frac{5}{6}, &8/9< F\leq 1.
\end{cases}
\end{equation}
for $q = 2$ and $d = 3$.

As depicted in Figure \ref{fig:figure22}, this lower bound in Eq.(\ref{eq:c22bound}) is shown to
be tight. Additionally, we compare the unified $(q,s)$-concurrence with the $q$-concurrences for the isotropic state at $q=2$. The value of $C_{2,2}(\rho_F)$
 is greater than $C_2(\rho_F)$.
\begin{figure*}[ht]
\centering
    \includegraphics[width=0.67\linewidth]{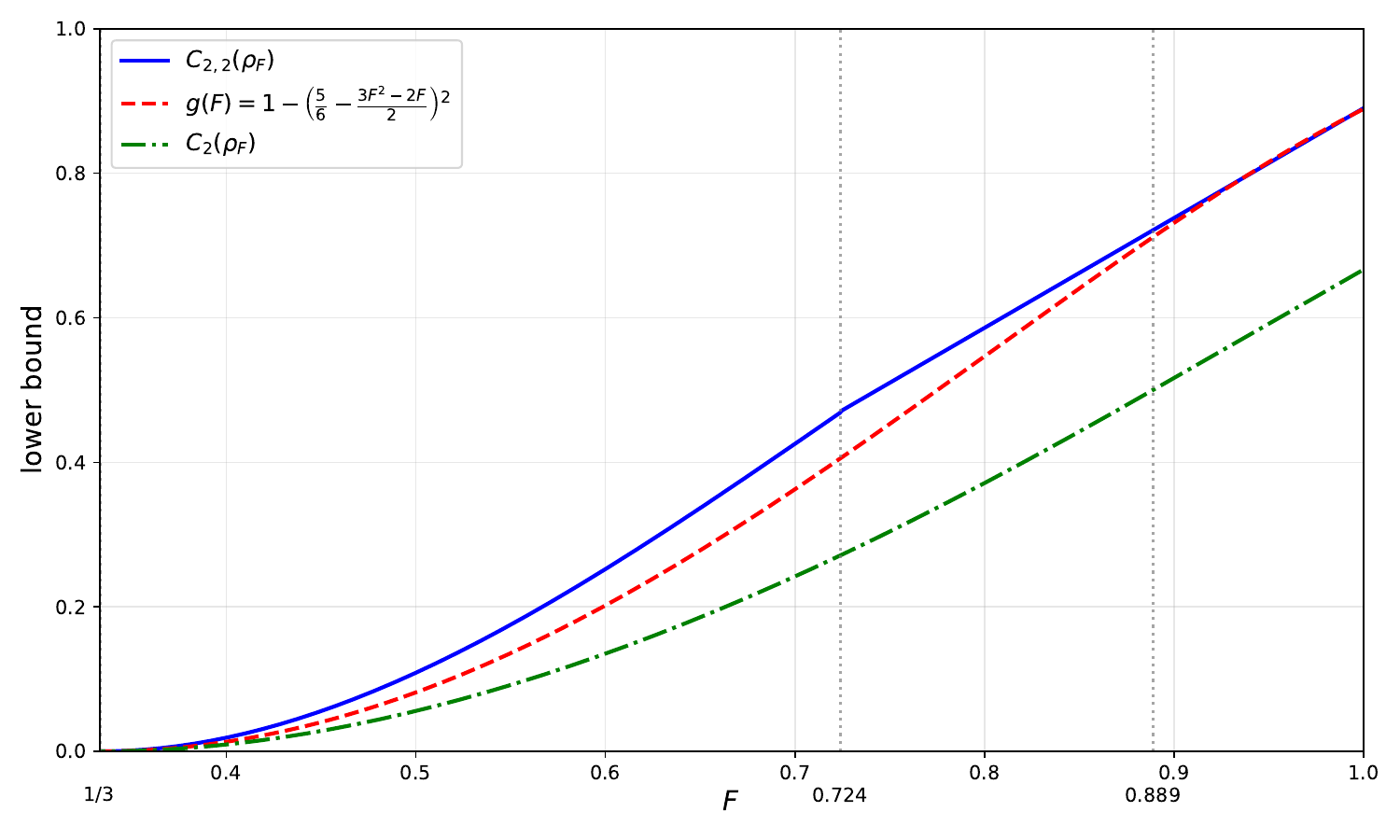}
     \caption{Entanglement measures for isotropic states: the exact value from Eq.(\ref{eq:c22function}) (blue solid); the lower bound from Eq.(\ref{eq:c22bound}) (red dashed); and the
$C_2$-concurrence from Eq.(\ref{eq:c2function}) (green dashed).}
    \label{fig:figure22}
\end{figure*}
\end{example}
%%%%%%%%%%%%%%%%%%%%%%%%%%%%%%%%%%%%%%%%%%%%%%%%%%%%%%%%%%%%%%%%%%%%%%
\begin{example}
For Werner states, the general expression is given by
\begin{eqnarray}
    \rho_{w}=\frac{2(1-w)}{d(d+1)} \left( \sum_{i=1}^{d} |ii\rangle \langle ii| + \sum_{i<k} |\Phi_{ik}^+ \rangle \langle \Phi_{ik}^+| \right)+ \frac{2w}{d(d-1)} \sum_{i<k} |\Phi_{ik}^- \rangle \langle \Phi_{ik}^-|,
\label{eq:werner_state}
\end{eqnarray}
where the states $|\Phi_{ik}^\pm \rangle$ are defined as $|\Phi_{ik}^\pm \rangle = (|ik\rangle \pm |ki\rangle)/\sqrt{2}$, with $|ik\rangle$ and $|ki\rangle$ representing standard basis states in the Hilbert space of the two subsystems.
The mixing parameter $w$ is given by the trace over the antisymmetric subspace:
\begin{equation*}
w = \tr\left(\rho_w \sum_{i<k} |\Phi_{ik}^- \rangle \langle \Phi_{ik}^-|\right),
\end{equation*}
as established in Ref.~\cite{lee2003convex}. The Werner state $\rho_w$ is separable if and only if $0 \leq w \leq \frac{1}{2}$~\cite{vollbrecht2001entanglement,werner1989quantum}.
For $w > \frac{1}{2}$, the state is entangled, and its entanglement measure can be explicitly expressed as (see Appendix~\ref{sec:Appendix B} for derivation):
\begin{equation}
\xi(\rho_w) =  1 - \left[\left( \frac{1+G}{2}\right)^q +\left( \frac{1-G}{2} \right)^q \right]^s,
\label{eq:entanglement_measure1}
\end{equation}
where $G = 2\sqrt{w(1-w)}$, $q>1$ and $qs\geq 1$. This expression quantifies the entanglement of the Werner state for $w > \frac{1}{2}$ and depends on both the mixing parameter $w$ and the exponents $q$ and $s$.

We now consider the specific case $d = 2$ and
$q = 3$ and $s=2$. In this case,  Eq.(\ref{eq:entanglement_measure1}) simplifies to
\begin{equation}
\xi(\rho_w) = \frac{3}{16} \left[(2w-1)^2 (8-3(2w-1)^2) \right].
\label{eq:entanglement_measure2}
\end{equation}
\begin{figure*}[ht]
\centering
    \includegraphics[width=0.65\linewidth]{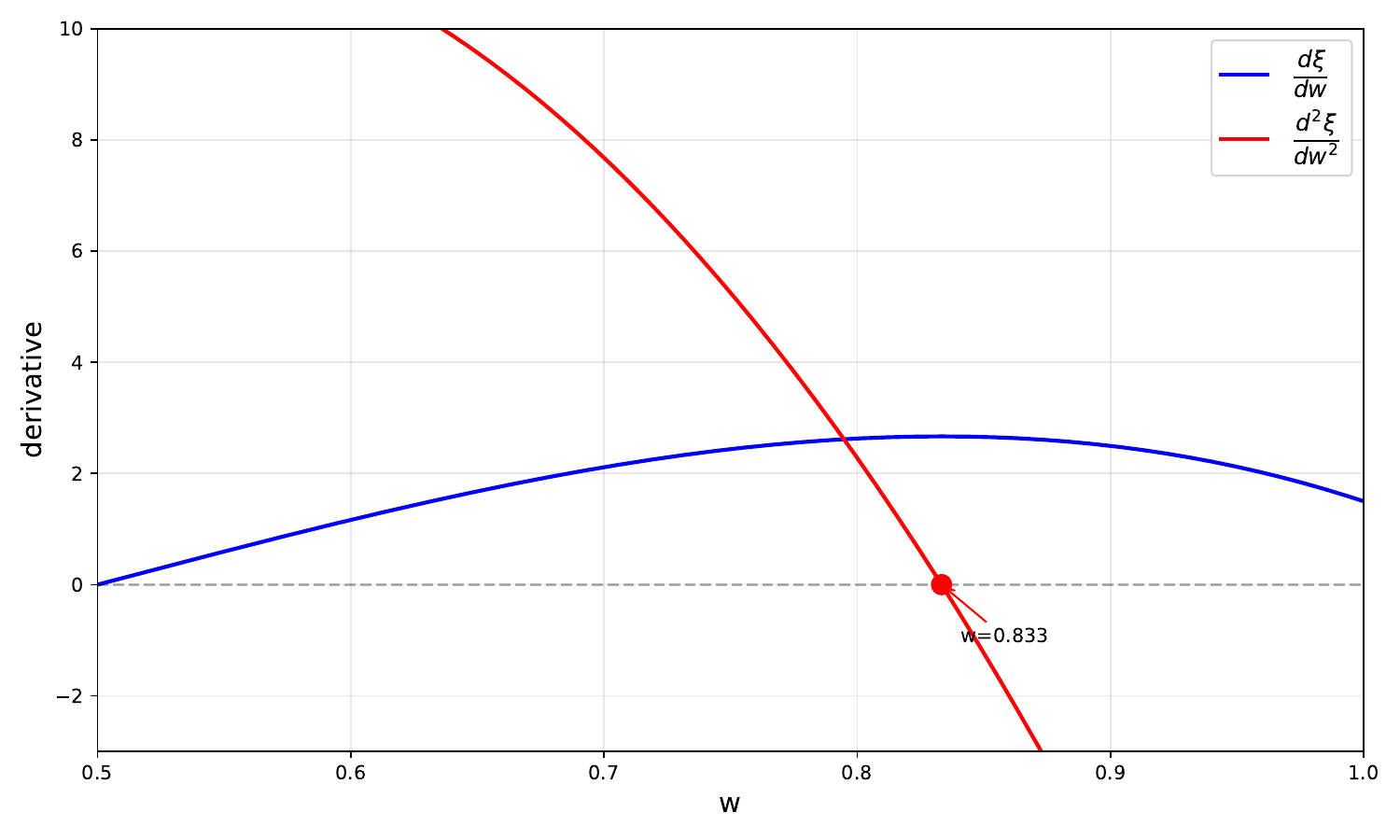}
     \caption{First and second derivatives of $\xi(\rho_w)$ with respect to $w$.}
    \label{fig:figure232}
\end{figure*}
Since \(\frac{d}{d w} \xi(\rho_w)>0\) for $w>\frac{1}{2}$, the function \(\xi(\rho_w)\) is monotonically increasing, as illustrated in Figure \ref{fig:figure232}. For \(w \geq 0.833\), we find \(\frac{d^2}{d w^2} \xi(\rho_w) < 0\), indicating that \(\xi(\rho_w)\) is not convex on \(w \in (0.833, 1]\).
As \(C_{3,2}(\rho_w)\) is defined as the largest convex function bounded above by Eq.(\ref{eq:qs_concurrence}), convexity is restored by linear interpolation between
the points at \(w = 0.833\) and \(w = 1\). Therefore,
we obtain the following result,
\begin{equation}
\label{eq:c32function}
    C_{3,2}(\rho_w) =
\begin{cases}
0, & w \leq 1/2, \\
\xi(\rho_w), & 1/2 < w \leq 0.833, \\
2.29w-1.35, & 0.833 < w \leq 1.
\end{cases}
\end{equation}

From Theorem \ref{th:bound1}, we obtain the following lower bound
\begin{equation}
\label{eq:c32bound}
    C_{3,2}(\rho_w)\geq \frac{5}{4}-\frac{5}{4}\left(\frac{1}{2}-2w^2+2w\right)^2.
\end{equation}

\begin{figure*}[ht]
\centering
    \includegraphics[width=0.67\linewidth]{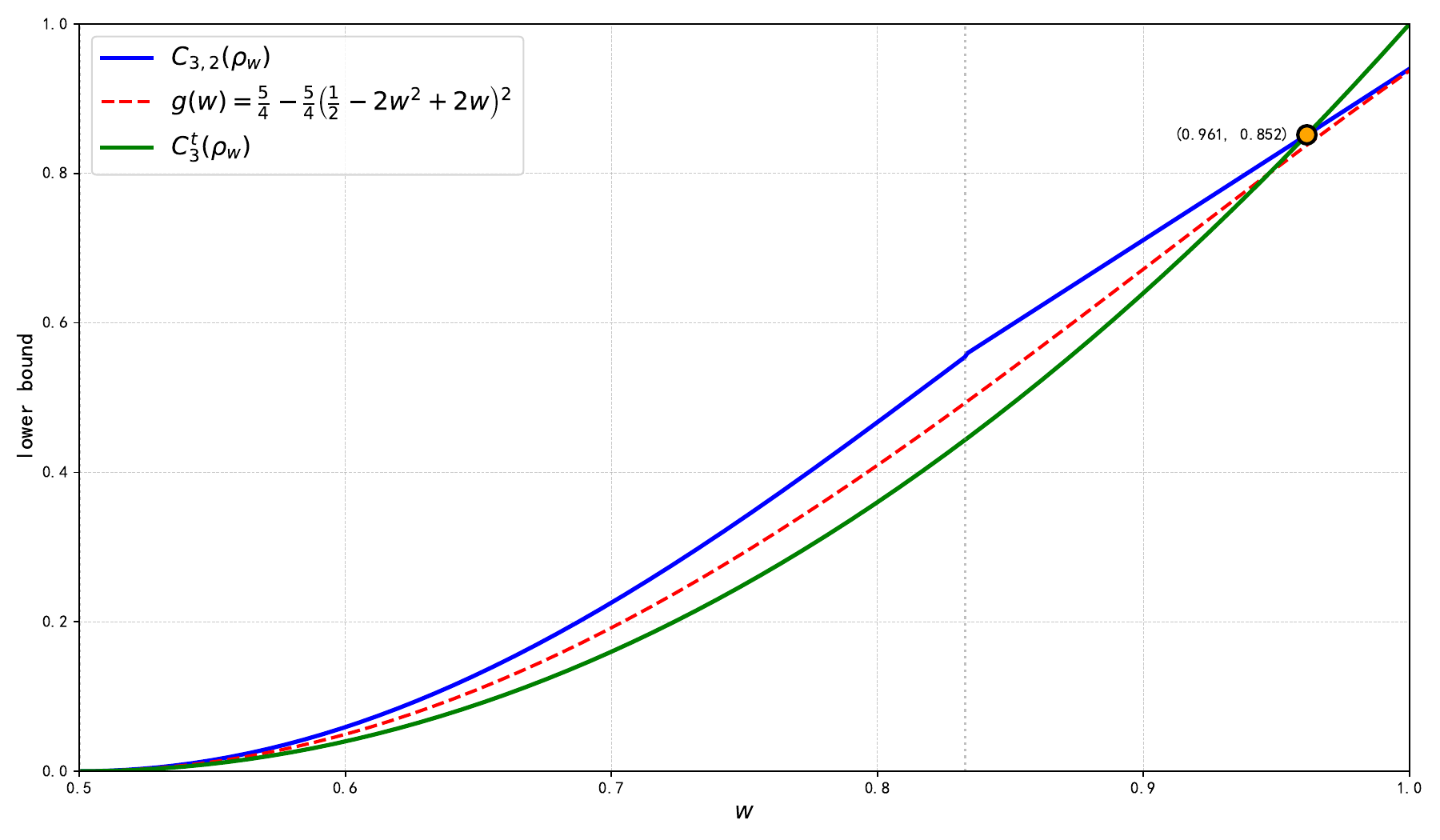}
     \caption{Entanglement measures for Werner states: the exact value from Eq.(\ref{eq:c32function}) (blue solid); the lower bound from Eq.(\ref{eq:c32bound}) (red dashed); and the
$C_3^t$-concurrence from Eq.~(\ref{eq:c3tfunction}) (green solid).}
    \label{fig:figurec232}
\end{figure*}

The concurrence of Werner states was derived in Ref.\cite{xuan2025new},
where it is given by
\begin{equation}
\label{eq:c3tfunction}
C_3^t (\rho_w) =
\begin{cases}
0, & w \leq \frac{1}{2}, \\
(2w - 1)^2, & w > \frac{1}{2}
\end{cases}
\end{equation}
for $q = 3$ and $d = 2$.

As shown in Figure \ref{fig:figurec232}, the lower bound in Eq.(\ref{eq:c32bound}) is tight.
Additionally, we compare the unified $(q,s)$-concurrence with the $C_q^t$-concurrences for Werner states at $q=3$. The value of $C_{3,2}(\rho_w)$
 is greater than $C_3^t(\rho_w)$
for $0.5\leq w \leq 0.961$.
\end{example}
%%%%%%%%%%%%%%%%%%%%%%%%%%%%%%%%%%%%%%%%%%%%%%%%%%%%%%%
\section{Entanglement inequalities based on the unified $(q,s)$-concurrence in
multipartite systems}
\label{sec:Monogamy}
\subsection{Monogamy of the unified $(q,s)$-concurrence}
The distribution of entanglement represents a significant aspect of entanglement theory and quantum information processing. Monogamy relations capture a fundamental property of entanglement, characterizing the constraints on how bipartite entanglement can be shared among multiple parties within a quantum system. The Coffman-Kundu-Wootters (CKW) inequality provides a quantitative formulation of this monogamy property~\cite{coffman2000distributed}.

Consider a pure state $|\phi\rangle_{AB}$ in the bipartite system $\mathbb{C}^2 \otimes \mathbb{C}^d$, where subsystem $A$ is a qubit ($\mathbb{C}^2$) and subsystem $B$ is a finite-dimensional Hilbert space ($\mathbb{C}^d$). Using the Schmidt decomposition, the state $|\phi\rangle_{AB}$ can be expressed as
\begin{equation*}
|\phi\rangle_{AB} = \sqrt{\lambda_0}|0\rangle|\phi_0\rangle + \sqrt{\lambda_1}|1\rangle|\phi_1\rangle,
\end{equation*}
where $\lambda_0$ and $\lambda_1$ are the Schmidt coefficients satisfying $\lambda_0 + \lambda_1 = 1$, and $\{|\phi_i\rangle\}$ are orthogonal states in subsystem $B$ (i.e., $\langle\phi_i|\phi_j\rangle = \delta_{ij}$).

From Eq.\eqref{eq:epsilonqs}, the normalized unified
$(q,s)$-concurrence for the pure state $|\phi\rangle_{AB}$ can be defined as
\begin{equation}
\hat{C}_{q,s}(|\phi\rangle_{AB}) = \frac{1}{\epsilon_{q,s}(1 - 2^{s(1-q)})}C_{q,s}(\rho_{A})=\frac{1-(\lambda_0^q+\lambda_1^q)^s}{ 1 - 2^{s(1-q)}},
\label{eq:concurrence_relation}
\end{equation}
where $\rho_A = \tr_B(|\phi\rangle_{AB}\langle\phi|)$ is the reduced density matrix of subsystem $A$. Note that the normalization factor $\frac{1}{\epsilon_{q,s}(1 - 2^{s(1-q)})}$ ensures proper scaling of the concurrence measure.

Furthermore, the concurrence for the pure state $|\phi\rangle_{AB}$ is given by \cite{rungta2001universal,hill1997entanglement}
\begin{equation}
C(|\phi\rangle_{AB}) = \sqrt{2(1 - \mathrm{tr}(\rho_A^2))} = 2\sqrt{\lambda_0 \lambda_1},
\end{equation}
where $\rho_A=\mathrm{tr}_B(|\phi\rangle_{AB}\langle\phi|)$ is the reduced density matrix. A one-to-one correspondence exists between the Schmidt coefficients $\lambda_0, \lambda_1$ and the concurrence $C(|\phi\rangle_{AB})$. The Schmidt coefficients can be expressed in terms of the concurrence as
\begin{equation*}
\lambda_{0,1} = \frac{1 \pm \sqrt{1 - C^2(|\phi\rangle_{AB})}}{2}.
\end{equation*}
From Eq.\eqref{eq:concurrence_relation}, it follows that
\begin{equation}
\hat{C}_{q,s}(|\phi\rangle_{AB})=h_{q,s}(C(|\phi\rangle_{AB})),
\label{eq:functional_relation2}
\end{equation}
where $h_{q,s}(x)$ is an analytic function defined as
\begin{eqnarray}
    \label{eq:hq_function}
    h_{q,s}(x)= \frac{1 - \left(\left(\frac{1 + \sqrt{1 - x^2}}{2}\right)^q
    + \left(\frac{1 - \sqrt{1 - x^2}}{2}\right)^q\right)^s}{1 - 2^{s(1-q)}}
    =\frac{2^{qs}+g_{q,s}(x)}{2^{qs}(1 - 2^{s(1-q)})},
\end{eqnarray}
where $g_{q,s}(x)=-\left[(1+\sqrt{1 - x^2})^q + (1-\sqrt{1 - x^2})^q\right]^s$
with $0 \leq x \leq 1$. Thus for any $\mathbb{C}^2 \otimes \mathbb{C}^d$ pure state $\ket{\phi}_{AB}$, we have a functional relation
between its concurrence and the normalized $(q,s)$-concurrence for $q$ and $s$ in Eq.(\ref{eq:epsilon_qs}).

\begin{theorem}
\label{th:functional_relation}
 For $q\geq1$, $0 \leq s \leq 1$ and $1\leq qs \leq 3$ and any qubit-qudit state $\rho_{AB}$,  the normalized unified $(q,s)$-concurrence $\hat{C}_{q,s}(\rho_{AB})$ is related to the concurrence $C(\rho_{AB})$ by
\begin{equation}
\hat{C}_{q,s}(\rho_{AB}) = h_{q,s}(C(\rho_{AB})),
\label{eq:functional_relation}
\end{equation}
where $h_{q,s}(x)$ is an analytic function defined in Eq.(\ref{eq:hq_function}).
\end{theorem}
\begin{proof}
For $q\geq1$, $0 \leq s \leq 1$ and $1\leq qs \leq 3$, the function $g_{q,s}(x)$ is monotonically increasing and convex on the interval $x \in (0,1)$~\cite{san2011unified}. This property implies that $h_{q,s}(x)$ is monotonically increasing and convex on the interval $x \in (0,1)$.

For a mixed state $\rho_{AB}$ in $\mathbb{C}^2 \otimes \mathbb{C}^d$,  there exists an optimal decomposition of $\rho_{AB}=\sum_i p_i |\phi_i\rangle_{AB}\langle\phi_i|$, where the pure states
$|\phi_i\rangle_{AB}$ achieve the convex-roof extension of the concurrence:
\begin{equation}
C(\rho_{AB}) = \sum_i p_i C(|\phi_i\rangle_{AB}).
\label{eq:concurrence_decomp}
\end{equation}

Using this decomposition, we can derive the following inequality:
\begin{eqnarray}
h_{q,s}\left(C(\rho_{AB})\right) &=& h_{q,s}\Big(\sum_i p_i C(|\phi_i\rangle_{AB})\Big) \nonumber \\
&=& \sum_i p_i h_{q,s}\Big(C(|\phi_i\rangle_{AB})\Big) \nonumber \\
&=& \sum_i p_i \hat{C}_{q,s}(|\phi_i\rangle_{AB}) \nonumber \\
&\geq& \hat{C}_{q,s}\left(\rho_{AB}\right),
\label{eq:Relationinequality1}
\end{eqnarray}
where the second equality follows from the linearity of the function $h_{q,s}$, and the inequality results from the definition of $\hat{C}_{q,s}$ as an infimum over all decompositions.

Conversely, consider the optimal decomposition $\rho_{AB}=\sum_j q_j |\mu_j\rangle_{AB}\langle\mu_j|$ for the normalized unified $(q,s)$-concurrence $\hat{C}_{q,s}(\rho_{AB})$. Therefore, we have that
\begin{eqnarray}
\label{eq:Relationinequality2}
\hat{C}_{q,s}(\rho_{AB}) &=& \sum_j q_j \hat{C}_{q,s}(|\mu_j\rangle_{AB}) \nonumber \\
&=& \sum_j q_j h_{q,s}(C(|\mu_j\rangle_{AB})) \nonumber \\
&\geq& h_{q,s}\Big(\sum_j q_j C(|\mu_j\rangle_{AB})\Big) \nonumber \\
&\geq& h_{q,s}(C(\rho_{AB})),
\end{eqnarray}
where the first inequality follows from the convexity of $h_{q,s}(x)$  and the second uses its monotonicity along with the fact that $\sum_j q_j C(|\mu_j\rangle_{AB})\geq C(\rho_{AB})$.

Combining inequalities~(\ref{eq:Relationinequality1}) and~(\ref{eq:Relationinequality2}), we conclude that
\begin{equation}
\hat{C}_{q,s}(\rho_{AB}) = h_{q,s}(C(\rho_{AB}))
\label{eq:main_result}
\end{equation}
for $q \geq 1$, $0 \leq s \leq 1$, $1\leq qs \leq 3$ and any mixed state $\rho_{AB}$.
\end{proof}
%%%%%%%%%%%%%%%%%%%%%%%%%%%%%%%%%%%%%%%%%%%%%%%%%%%%%%
It has been demonstrated that the concurrence follows
a monogamy relation for any $n$-qubit state $\rho_{A_1 A_2 \cdots A_n}$ \cite{osborne2006general}, expressed as:
\begin{equation}
\label{eq:monogamyc2}
  C^2(\rho_{A_1|A_2 \cdots A_n}) \geq C^2(\rho_{A_1|A_2}) + \cdots + C^2(\rho_{A_1| A_n}).
\end{equation}
where \(C^2(\rho_{A_1|A_2 \cdots A_n})\) denotes the concurrence with respect to the bipartition between subsystem \(A_1\) and the remaining subsystems
\(A_2, \cdots, A_n\), and \( C^2(\rho_{A_1|A_i})\) is the concurrence of the reduced density matrix \(\rho_{A_1|A_i}\) for each \(i = 2, \cdots, n\). In this work, we demonstrate that this entanglement monogamy can also be characterized using the normalized unified $(q,s)$-concurrence.

\begin{theorem}
  For any multi-qubit state \(\rho_{A_{1} A_{2} \cdots A_{n}}\) and parameters satisfying \(q \geq 2\), \(0 \leq s \leq 1\), and \(1\leq qs \leq 3\), the following inequality holds:
  \begin{equation}
  \label{eq:inequalitymonogamy}
      \hat{C}_{q,s}\left( \rho_{A_1|A_2 \cdots A_n} \right) \geqslant \hat{C}_{q,s}( \rho_{A_1| A_2} ) + \cdots + \hat{C}_{q,s}( \rho_{A_1| A_n} ),
  \end{equation}
where \(\hat{C}_{q,s}\left( \rho_{A_1 |A_2 \cdots A_n} \right)\) denotes the normalized unified $(q,s)$-concurrence across the bipartition \(A_1\) versus \(A_2 \cdots A_n\), and \(\hat{C}_{q,s}( \rho_{A_1|A_i} )\) is the normalized unified $(q,s)$-concurrence of
the reduced state \(\rho_{A_1| A_i}\) for each \(i = 2, \cdots, n\).
\end{theorem}
\begin{proof}
We begin by proving the theorem for an \(n\)-qubit pure state \(|\psi\rangle_{A_1 A_2 \cdots A_n}\). Note that inequality (\ref{eq:monogamyc2}) implies
\begin{equation}
\label{eq:monogamyc3}
  C(\rho_{A_1|A_2 \cdots A_n}) \geq \sqrt{C^2(\rho_{A_1|A_2}) + \cdots + C^2(\rho_{A_1| A_n})} .
\end{equation}
for any such pure state.

It has been shown \cite{san2011unified} that for \(q \geq 2\), \(0 \leq s \leq 1\) and \(1\leq qs \leq 3\),
\begin{equation}
\label{eq:hqsxy}
    h_{q,s}(\sqrt{x^2+y^2})\geq h_{q,s}(x)+h_{q,s}(y),
\end{equation}
on the domain $D = \{(x,y)|0 \leq x,y,x^2 + y^2 \leq 1\}$.

From (\ref{eq:monogamyc3}) and (\ref{eq:hqsxy}), it follows that
\begin{eqnarray}
&&\hat{C}_{q,s}\left( |\psi\rangle_{A_1|A_2 \cdots A_n} \right)\nonumber\\
&\overset{(a)}{=}& h_{q,s} \left( C(|\psi\rangle_{A_1|A_2 \cdots A_n}) \right)\nonumber \\
&\overset{(b)}{\geq}& h_{q,s} \left( \sqrt{C^2(\rho_{A_1|A_2}) + \cdots + C^2(\rho_{A_1| A_n})} \right)\nonumber \\
&\geq& h_{q,s} \left( C(\rho_{A_1|A_2}) \right) \nonumber\\
&&+ h_{q,s} \left( \sqrt{C^2(\rho_{A_1|A_3}) + \cdots + C^2(\rho_{A_1| A_n})} \right) \nonumber\\
&&\;\;\vdots \nonumber\\
&\geq& h_{q,s} \left(  C(\rho_{A_1|A_2}) \right) + \cdots + h_{q,s} \left(  C(\rho_{A_1|A_n}) \right) \nonumber\\
&\overset{(c)}{=}& \hat{C}_{q,s}( \rho_{A_1| A_2} ) + \cdots + \hat{C}_{q,s}( \rho_{A_1| A_n} ),
\end{eqnarray}
where (a) follows from the functional relation between the concurrence and  normalized unified $(q,s)$-concurrence for pure states in a \(\mathbb{C}^2 \otimes \mathbb{C}^d\) system; (b) follows from the monotonicity of \(h_{q,s}(x)\), while the subsequent inequalities are obtained
through iterative application of (\ref{eq:hqsxy}); and the final equality (c) is established by Theorem \ref{th:functional_relation}.

Now consider an \(n\)-qubit mixed state \(\rho_{A_1 A_2 \cdots A_n}\). Let \(\rho_{A_1 |A_2 \cdots A_n} = \sum_j p_j |\psi_j\rangle_{A_1 |A_2 \cdots A_n} \langle \psi_j|\) be an optimal decomposition achieving
\begin{equation}
  \hat{C}_{q,s} \left( \rho_{A_1|A_2 \cdots A_n} \right) = \sum_j p_j  \hat{C}_{q,s} \left( |\psi_j\rangle_{A_1| A_2 \cdots A_n} \right).
\end{equation}

Since each \(|\psi_j\rangle\) is an \(n\)-qubit pure state, we have that
\begin{eqnarray}
 &&\hat{C}_{q,s} \left( \rho_{A_1 |A_2 \cdots A_n} \right)\nonumber\\
 &=& \sum_j p_j \hat{C}_{q,s} \left( |\psi_j\rangle_{A_1 |A_2 \cdots A_n} \right) \nonumber\\
&\geq& \sum_j p_j \left[ \hat{C}_{q,s}( \rho_j^{A_1| A_2} ) + \cdots + \hat{C}_{q,s}( \rho_j^{A_1| A_n} ) \right] \nonumber\\
&=& \sum_j p_j \hat{C}_{q,s}( \rho_j^{A_1 A_2} ) + \cdots + \sum_j p_j \hat{C}_{q,s}( \rho_j^{A_1 A_n} )\nonumber \\
&\geq& \hat{C}_{q,s}( \rho_{A_1| A_2} ) + \cdots + \hat{C}_{q,s}( \rho_{A_1| A_n} ),
\end{eqnarray}
where the last inequality follows from the definition of the normalized unified $(q,s)$-concurrence for each reduced state \(\rho_{A_1 A_i}\).
\end{proof}

Next, we provide an example to demonstrate that the normalized unified $(q,s)$-concurrence satisfies to the monogamy property.
\begin{example}
Consider the three-qubit state \(|\phi\rangle_{ABC}\) given in the generalized Schmidt decomposition form \cite{acin2000generalized,gao2008estimation}:
\begin{equation}
    |\phi\rangle_{ABC} = \lambda_0 |000\rangle + \lambda_1 e^{i\varphi} |100\rangle + \lambda_2 |101\rangle + \lambda_3 |110\rangle + \lambda_4 |111\rangle,
\end{equation}
with \(\lambda_i \geq 0\) for \(i=0,1,\cdots, 4\), and normalization \(\sum_{i=0}^4 \lambda_i^2 = 1\). The concurrences in this system are
\begin{eqnarray}
    C(\ket{\phi}_{A|BC}) &=&  \sqrt{4\lambda_0^2(1-\lambda_0^2-\lambda_1^2)},\\
    C(\ket{\phi}_{A|B}) &=& 2\lambda_0 \lambda_2,\\
    C(\ket{\phi}_{A|C}) &=& 2\lambda_0 \lambda_3.
\end{eqnarray}

For \(q \geq 2\), \(0 \leq s \leq 1\) and \(1\leq qs \leq 3\), we have that
    \begin{eqnarray}
\label{eq:functionalnABC}
K&=&\hat{C}_{q,s}(|\phi\rangle_{A|BC})\nonumber \\
&=& \frac{2^{qs} -\left[\Big(1 + \sqrt{1 - 4\lambda_0^2(1-\lambda_0^2-\lambda_1^2)}\Big)^q
    + \Big(1 - \sqrt{1 - 4\lambda_0^2(1-\lambda_0^2-\lambda_1^2)}\Big)^q\right]^s}{2^{qs}(1 - 2^{s(1-q)})},\\
\label{eq:functionalnAB}
K_1&=&\hat{C}_{q,s}(|\phi\rangle_{A|B}) \nonumber \\
&=& \frac{2^{qs} -\left[\Big(1 + \sqrt{1 - 4\lambda_0^2\lambda_2^2}\Big)^q
    + \Big(1 - \sqrt{1 - 4\lambda_0^2\lambda_2^2}\Big)^q\right]^s}{2^{qs}(1 - 2^{s(1-q)})},\\
\label{eq:functionalnAC}
K_1&=&\hat{C}_{q,s}(|\phi\rangle_{A|C})\nonumber\\
&=& \frac{2^{qs} -\left[\Big(1 + \sqrt{1 - 4\lambda_0^2\lambda_3^2}\Big)^q
    + \Big(1 - \sqrt{1 - 4\lambda_0^2\lambda_3^2}\Big)^q\right]^s}{2^{qs}(1 - 2^{s(1-q)})}.
\end{eqnarray}

In order to  verify the inequality (\ref{eq:inequalitymonogamy}), we set the values \(\lambda_0 =\sqrt{2/7}, \lambda_1 = \lambda_2 = \sqrt{1/7}\) and \(\lambda_3 = \sqrt{3/7}\). Figure \ref{fig:monogamy} shows the residual entanglement
\begin{eqnarray}
    \tau_{\hat{C}_{q,s}}(|\phi\rangle_{ABC})=K-K_1 -K_2.
    \label{eq:rentanglement1}
\end{eqnarray}
\begin{figure}[ht]
\centering
    \includegraphics[width=1.05\linewidth]{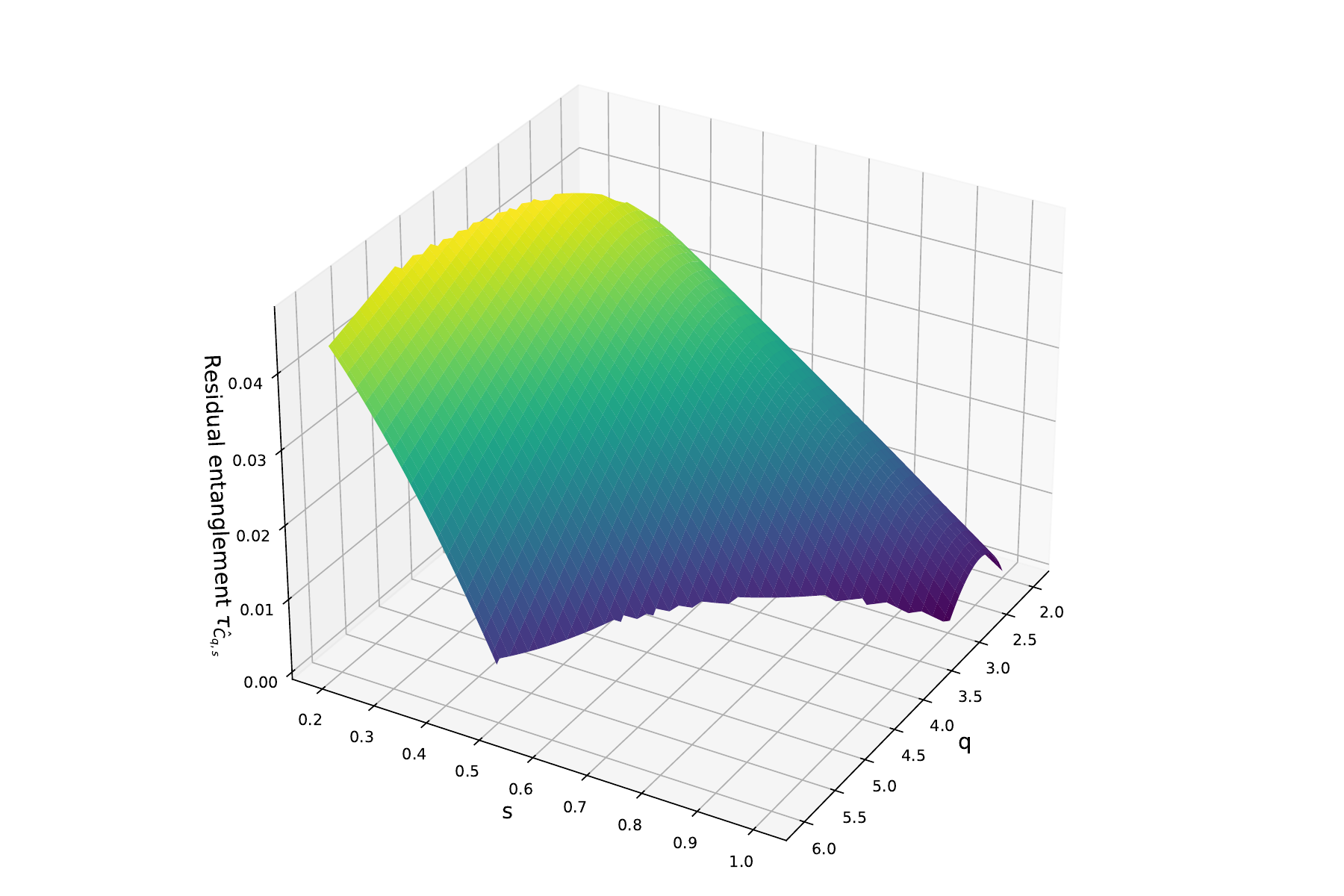}
     \caption{The residual entanglement $\tau_{\hat{C}_{q,s}}=K-K_1 -K_2$ of the normalized unified $(q,s)$-concurrence in Eq.(\ref{eq:rentanglement1}).}
    \label{fig:monogamy}
\end{figure}
\begin{figure*}[ht]
\centering
    \includegraphics[width=0.65\linewidth]{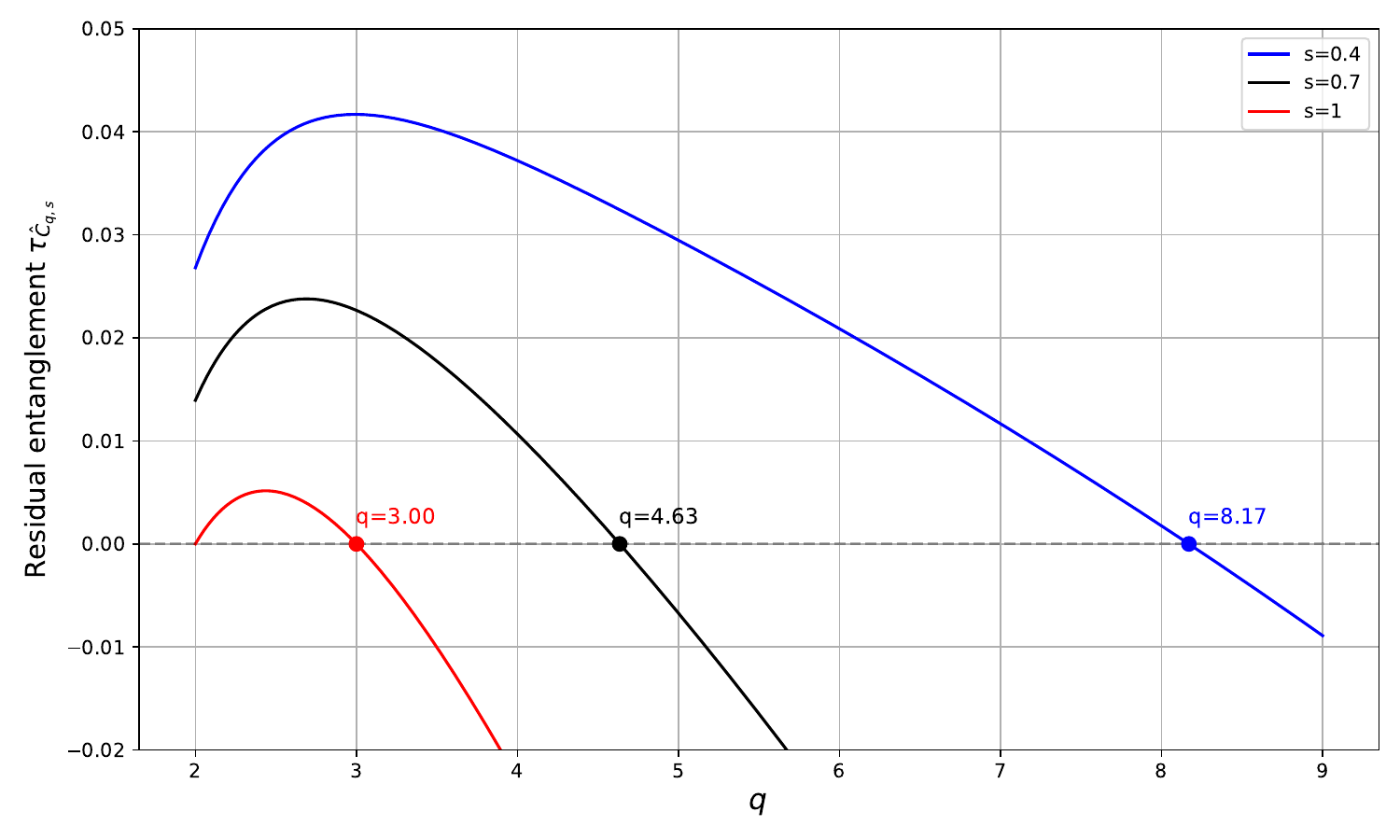}
     \caption{The residual entanglement $\tau_{\hat{C}_{q,s}}$ of the normalized unified $(q,s)$-concurrence, where $s=1$, $s=0.7$ and $s=0.4$.}
    \label{fig:monogamyS3}
\end{figure*}

As shown in Figure \ref{fig:monogamyS3}, when $s$ decreases, the $q$ satisfying the monogamy for this state will be larger. When $s=1$, the residual entanglement $\tau_{\hat{C}_{q,1}}$ is the same as that given by the \(q\)-concurrence, i.e., $\tau_{\hat{C}_{q,1}}(|\phi\rangle_{ABC})=\tau_{\hat{C}_{q}}(|\phi\rangle_{ABC})$. In addition, the monogamy of the
\(q\)-concurrence does not work for this state when $q\geq 3$. If $s=0.4$, then
the monogamy of the  unified $(q,s)$-concurrence still works for state when $2\leq q \leq 8.17$.
\end{example}
%%%%%%%%%%%%%%%%%%%%%%%%%%%%%%%%%%%%%%%%%%%%%%%%%%%%%%%%%%%%%%%%%%%%%%
\subsection{Entanglement polygon inequalities for the unified $(q,s)$-concurrence}

In this section, we discuss the entanglement polygon inequalities for the unified $(q,s)$-concurrence. Unlike monogamy inequalities, entanglement polygon inequalities provide another class of mathematical relations that capture the intricate correlations among particles in multipartite systems \cite{qian2018entanglement}.

An $n$-qudit pure state defined on the Hilbert space $\otimes_{i=1}^n \mathcal{H}_i$ can be generally expressed as
\begin{equation}
  |\psi\rangle_{1 \cdots n} = \sum_{s_1=0}^{d_1-1} \cdots \sum_{s_n=0}^{d_n-1} \alpha_{s_1 \cdots s_n} |s_1 \cdots s_n\rangle,
  \label{eq:psi1n}
\end{equation}
where $\dim(\mathcal{H}_i) = d_i$ with $d_i \geq 2$ for $i = 1, \cdots, n$, and the complex coefficients $\alpha_{s_1 \cdots s_n}$ satisfy the normalization condition $\sum_{s_1, \cdots, s_n} |\alpha_{s_1 \cdots s_n}|^2 = 1$.

Firstly, we consider a one-to-group entanglement measure $\mathcal{E}^{j|\overline{j}}$ for the pure state given in Eq.(\ref{eq:psi1n}), where the index $j$ denotes a single subsystem, and $\overline{j}$ represents all the remaining subsystems. Here, $\mathcal{E}$ is an arbitrary bipartite entanglement measure. Such bipartite entanglements are also referred to as quantum marginal entanglements~\cite{qian2018entanglement}. Based on the unified $(q,s)$-concurrence, we derive the following relationship among all one-to-group marginal entanglement measures for an arbitrary $n$-qudit system.

\begin{theorem}
\label{eq:polygon1}
For any $n$-qudit pure state $|\psi\rangle$ on the Hilbert space $\otimes_{i=1}^n \mathcal{H}_i$, the following inequality holds,

\begin{equation}
  C^{{j|\overline{j}}}_{q,s}(|\psi\rangle) \leq \sum_{\substack{ k\neq j,\ \forall k
  }}C_{q,s}^{k|\overline{k}}(|\psi\rangle)
\end{equation}
 for any $q\geq 1$ and $qs\geq 1$.
\end{theorem}
\begin{proof}
    For any $q\geq 1$ and $qs\geq 1$, i.e., $\epsilon_{q,s}=1$, it means that
    \begin{equation}
        C_{q,s}(\ket{\psi}_{AB})=1-(\tr(\rho_A^q))^s=F_{q,s}(\rho_A).
        \label{eq:cqsfqs}
    \end{equation}
 Therefore, we have that
\begin{eqnarray}
    C_{q,s}^{j|\overline{j}}(|\psi\rangle)&=&F_{q,s}(\rho_j)\nonumber\\
    &=&F_{q,s}(\rho_{\overline{j}})\leq \sum_{\substack{ k\neq j,\ \forall k}}F_{q,s}(\rho_{k})\label{eq:ploygon2}\\
    &=&\sum_{\substack{ k\neq j,\ \forall k}}C_{q,s}^{k|\overline{k}}(|\psi\rangle)\label{eq:ploygon3},
\end{eqnarray}
where the inequality (\ref{eq:ploygon2}) is followed by iteratively using the
subadditivity condition in Eq.(\ref{eq:Subadditivity}), and $F_{q,s}(\rho_j)$ is the entropy of the reduced density matrix $\rho_j$ of the subsystem $j$;
the inequality (\ref{eq:ploygon3}) is obtained from Eq.(\ref{eq:cqsfqs}).
\end{proof}

Next, we consider an $(m + n)$-partite system described by the state $|\psi\rangle$ on the Hilbert space $\otimes_{i=1}^m \mathcal{H}_{A_i} \otimes_{j=1}^n \mathcal{H}_{B_j}$. Let $\textbf{A} = A_1 \cdots A_m$ and $\textbf{B} = B_1\cdots B_n$, where each $A_i$ and $B_j$ represents a qudit subsystem. For the state $|\psi\rangle^{\textbf{AB}} = |\psi\rangle^{A_1 \cdots A_m |B_1 \cdots B_n}$, We establish the following polygon inequality for bipartite high-dimensional entanglement.

\begin{theorem}
For any entangled pure state $|\psi\rangle_{\textbf{AB}}$, $q\geq 1$ and $qs\geq 1$, the unified $(q,s)$-concurrence satisfies the inequality
\begin{equation}
    C_{q,s}^{\textbf{A}|\textbf{B}}(|\psi\rangle) \leq \sum_{i=1}^m C_{q,s}^{A_i|\overline{A}_i}(|\psi\rangle),
\end{equation}
where $C_{q,s}^{\textbf{A}|\textbf{B}}$ denotes the the unified $(q,s)$-concurrence with respect to the bipartition $\textbf{A}|\textbf{B}$, and $C_{q,s}^{A_i|\overline{A}_i}$ represents the marginal entanglement between subsystem $A_i$ and the remaining subsystems $\overline{A}_i$.
\end{theorem}

\begin{proof}
    For any $q\geq 1$ and $qs\geq 1$, i.e., $\epsilon_{q,s}=1$, it means that
    \begin{equation}
        C_{q,s}(\ket{\psi})=1-(\tr(\rho_A^q))^s=F_{q,s}(\rho_A).
        \label{eq:cqsfqs1}
    \end{equation}
Furthermore, we obtain that
\begin{eqnarray}
    C_{q,s}^{\textbf{A}|\textbf{B}}(|\psi\rangle)&=&F_{q,s}(\rho_\textbf{A})\nonumber\\
    &\leq& \sum_{i=1}^mF_{q,s}(\rho_{A_i})\label{eq:ploygon23}\\
    &=&\sum_{i=1}^mC_{q,s}^{A_i|\overline{A}_i}(|\psi\rangle).\label{eq:ploygon24}
\end{eqnarray}
Similar to the inequalities (\ref{eq:ploygon2}) and (\ref{eq:ploygon3}),  the inequalities (\ref{eq:ploygon23}) and (\ref{eq:ploygon24}) hold for any $(m+n)$-qudit pure state $\ket{\psi}_{\textbf{AB}}$ on the Hilbert space $\otimes_{i=1}^m \mathcal{H}_{A_i} \otimes_{j=1}^n \mathcal{H}_{B_j}$.
\end{proof}
%%%%%%%%%%%%%%%%%%%%%%%%%%%%%%%%%%%%%%%%%%%%%%%%%%%%%%%%%%
\section{Conclusion}
\label{sec:Conclusion}
In summary, we have proposed and investigated a comprehensive two-parameter  framework for quantifying bipartite quantum entanglement: the unified
$(q,s)$-concurrence. This measure generalizes and unifies several earlier proposals, reducing to the standard concurrence for particular parameter choices and encompassing various parameterized concurrences as special cases. The two-parameter structure provides a more flexible characterization of quantum entanglement. By contrast, the two parameters can independently control different statistical contributions, allowing the measure to distinguish weakly entangled states, strongly entangled states, and mixed-state structures more effectively.
Furthermore, we demonstrated that it constitutes a well-defined entanglement measure. Analytical lower bounds for this concurrence have been derived for general mixed states based on the PPT and realignment criteria. Explicit expressions have been obtained for the unified $(q,s)$-concurrence for isotropic states and Werner states. In addition, we have investigated the monogamy properties of the unified $(q,s)$-concurrence for \(q \geq 2\), \(0 \leq s \leq 1\), and \(1\leq qs \leq 3\), in qubit systems.  This indicates that the two-parameter framework not only preserves the important properties of previously known one-parameter results, but also extends them to a more general entanglement-sharing framework. Therefore, the present results may provide useful tools for studying entanglement distribution in multipartite quantum systems, resource constraints in quantum networks, and security analysis in quantum communication protocols. At last, we
derived an entanglement polygon inequality for the unified $(q,s)$-concurrence with $q\geq 1$ and $qs\geq 1$, which manifests the relationship among all marginal entanglements in any multipartite qudit system.
Our work introduces a versatile and unifying framework for entanglement quantification. The unified
$(q,s)$-concurrence, with its adjustable parameters, offers a powerful tool for probing entanglement structures across different quantum states and systems.

Generalized deformation frameworks have been widely applied in weak measurements \cite{mousavigharalari2025q} and quantum Fisher information \cite{hasegawa2006quantum}. Our unified $(q,s)$-concurrence may provide a basis for extending two-parameter deformations to other quantum-information quantities and exploring their roles in quantum correlations, resource theory, and noisy quantum systems.

%%%%%%%%%%%%%%%%%%%%%%%%%%%%%%%%%%%%%%%%%%%%%%%%%%%%%%%%%%
\section*{Acknowledgments}
We want to express our gratitude to anonymous referees for their valuable and constructive comments.
\section*{Author Contributions}
Chen-Ming Bai wrote and revised the text of the main manuscript, Yu Luo participated in the discussion and revision.
\section*{Funding}
This research was supported by the National Natural Science Foundation of China under Grant No.12301590 and Science Research Project of Hebei Education Department under Grant No.BJ2025061.
\section*{Data Availability}
Data Availability No datasets were generated or analysed during the current study.
\section*{Declarations}
\textbf{Competing interests} The authors declare no competing interests.

\appendix
\section{Proof of the Lemma \ref{le:concave}}
\label{sec:Appendix}
\textit{Proof of the Lemma \ref{le:concave}:}

(1) \textbf{Concavity.}
    For $\sum_{i=1}^np_iF_{q,s}(\rho_i)\leq F_{q,s}(\rho)$, we first prove $\lambda F_{q,s}(\rho)+\mu F_{q,s}(\sigma)\leq F_{q,s}(\lambda\rho+\mu\sigma)$ with $\lambda,\mu\geq 0$ and $\lambda+\mu=1$.
    Here, from Minkowski's inequality \cite{1985The} with positive semidefinite matrices $\rho$ and $\sigma$, we get
\begin{equation}
\label{eq:Minkowski}
[\tr(\rho + \sigma)^q]^{1/q} \leq [\tr(\rho)^q]^{1/q} + [\tr(\sigma)^q]^{1/q}
\end{equation}
for \(q \geq 1\).

Case 1: when \(q \geq 1\), from Eq.(\ref{eq:Minkowski}), we have that
\begin{equation}
\label{eq:Minkowski1}
[\tr(\lambda \rho + \mu \sigma)^q]^{1/q} \leq \lambda [\tr(\rho)^q]^{1/q} + \mu [\tr(\sigma)^q]^{1/q},
\end{equation}
where \(\lambda \geq 0, \mu \geq 0, \lambda + \mu = 1\). Due to \(qs\geq  1\), we have
\begin{eqnarray}
\label{eq:Minkowski2}
[\tr(\lambda \rho + \mu \sigma)^q]^s &\leq& [\lambda (\tr(\rho)^q)^{1/q} + \mu (\tr(\sigma)^q)^{1/q}]^{qs}\nonumber\\
&\leq& \lambda [\tr(\rho)^q]^s + \mu [\tr(\sigma)^q]^s,
\end{eqnarray}
where Eq.(\ref{eq:Minkowski2}) is obtained from the convexity of the function \(y = x^{qs}\) for \( qs \geq 1\).  The inequality Eq.(\ref{eq:Minkowski2}) implies that
\begin{equation}
\label{eq:convave1}
1-[\tr(\lambda \rho + \mu \sigma)^q]^s \geq
 \lambda(1- [\tr(\rho)^q]^s) + \mu(1- [\tr(\sigma)^q]^s).
\end{equation}
Therefore, we have that
$$\lambda F_{q,s}(\rho)+\mu F_{q,s}(\sigma)\leq F_{q,s}(\lambda\rho+\mu\sigma).$$
By induction on $i$, we obtain the following inequality:
\begin{equation*}
    \sum_{i=1}^np_iF_{q,s}(\rho_i)\leq F_{q,s}(\rho),
\end{equation*}
where $\{p_i\}$ is the probability distribution corresponding to density operators $\rho_i$ of $\rho$. The equality holds iff all the states $\rho_i$ are identical.

Hence, $F_{q,s}(\rho)=1-(\tr\rho^q)^s$ is a concave function for $q\geq 1$ and $qs\geq1$.

Case 2: when \(0 < q < 1\), the inequality in Eq.(\ref{eq:Minkowski}) is reversed,
\begin{equation}
\label{eq:Minkowski4}
[\tr(\rho + \sigma)^q]^{1/q} \geq  [\tr(\rho)^q]^{1/q} +  [\tr(\sigma)^q]^{1/q},
\end{equation}

Since \(0 < qs < 1\), we have
\begin{eqnarray}
\label{eq:Minkowski5}
[\tr(\lambda \rho + \mu \sigma)^q]^s &\geq& [\lambda (\tr(\rho)^q)^{1/q} + \mu (\tr(\sigma)^q)^{1/q}]^{qs}\nonumber\\
&\geq& \lambda [\tr(\rho)^q]^s + \mu [\tr(\sigma)^q]^s,
\end{eqnarray}
where \(\lambda \geq 0, \mu \geq 0, \lambda + \mu = 1\).
The inequality Eq.(\ref{eq:Minkowski5}) implies that
\begin{equation}
\label{eq:convave2}
[\tr(\lambda \rho + \mu \sigma)^q]^s -1\geq
 \lambda( [\tr(\rho)^q]^s-1) + \mu([\tr(\sigma)^q]^s-1).
\end{equation}
Therefore, we have that
$$\lambda F_{q,s}(\rho)+\mu F_{q,s}(\sigma)\leq F_{q,s}(\lambda\rho+\mu\sigma).$$
By induction on $i$, we obtain the following inequality:
\begin{equation*}
    \sum_{i=1}^np_iF_{q,s}(\rho_i)\leq F_{q,s}(\rho),
\end{equation*}
where $\{p_i\}$ is the probability distribution corresponding to density operators $\rho_i$ of $\rho$. The equality holds iff all the states $\rho_i$ are identical.

So  $F_{q,s}(\rho)=(\tr\rho^q)^s-1$ is a concave function for \(0 < q < 1, 0 < qs < 1\).

\textbf{(2) Subadditivity.} Since $q\geq 1$ and $qs\geq1$, then the $E^{(s)}_q(\rho_{AB})=\frac{1}{(q-1)s}[1-(\tr(\rho^q))^s]$ in Ref.\cite{rastegin2011some} satisfies that
\begin{equation}
    E^{(s)}_q(\rho_{AB})\leq E^{(s)}_q(\rho_{A})+E^{(s)}_q(\rho_{B}).
    \label{eq:Esq}
\end{equation}
From Eq.(\ref{eq:Esq}), we have that
\begin{equation}
    1-[\tr(\rho_{AB}^q)]^s\leq1-[\tr(\rho_{A}^q)]^s+1-[\tr(\rho_{B}^q)]^s.
\end{equation}
It means that
\begin{equation}
    F_{q,s}(\rho_{AB})\leq F_{q,s}(\rho_{A})+F_{q,s}(\rho_{B}).
    \label{eq:Fqs}
\end{equation}

Next, we prove that $F_{q,s}(\rho_{AB})\geq |F_{q,s}(\rho_{A})-F_{q,s}(\rho_{B})|$. Given a bipartite pure state $\ket{\psi}_{ABC}$, from the Schmidt decomposition of $\ket{\psi}_{ABC}$, the density matrices $\rho_{AB}$ and $\rho_C$ have the same nonzero eigenvalues. Hence $F_{q,s}(\rho_{AB})= F_{q,s}(\rho_{C})$. Similarly, we have $F_{q,s}(\rho_{A})= F_{q,s}(\rho_{BC})$. Combining these  with the inequality (\ref{eq:Fqs}), we obtain
\begin{equation}
    F_{q,s}(\rho_{A})-F_{q,s}(\rho_{B})\leq F_{q,s}(\rho_{AB}).
    \label{eq:Fqs1}
\end{equation}
By symmetry, we also have that
\begin{equation}
    F_{q,s}(\rho_{B})-F_{q,s}(\rho_{A})\leq F_{q,s}(\rho_{AB}).
    \label{eq:Fqs2}
\end{equation}

Combining Eqs.(\ref{eq:Fqs1}) and (\ref{eq:Fqs2}), we can have the claim that
\begin{equation}
    |F_{q,s}(\rho_{A})-F_{q,s}(\rho_{B})|\leq F_{q,s}(\rho_{AB}).
    \label{eq:Fqs3}
\end{equation}
%%%%%%%%%%%%%%%%%%%%%%%%%%%%%%%%%%%%%%%%%%%%%%%%%%%%
\section{ Analytical Expression of $\xi(\rho_F)$ for
Isotropic States} \label{sec:Appendix A}
The proof of Lemma \ref{lem:isotropic states} draws inspiration from recent advances in the application of local symmetry~\cite{vollbrecht2001entanglement,wang2016entanglement,lee2003convex}. For the symmetric state $\rho_F$, the unified $(q,s)$-concurrence is given by

\begin{equation}
C_{q,s}(\rho_F) = \mathrm{co}[\xi(F, q,s, d)],
\label{eq:Aqs_concurrence}
\end{equation}
where the function $\xi(F, q,s, d)$ is defined as
\begin{equation}
\xi(F,q,s,d) = \inf \Bigl\{ C_{q,s}(|\psi\rangle) \Bigm| f_{\Psi^+}(|\psi\rangle) = F,\ \operatorname{rank}(\rho_\Psi) \leqslant d \Bigr\},
\label{eq:Axi_definition}
\end{equation}
and $\operatorname{rank}(\rho_\Psi)$ denotes the rank of the density operator $\rho_\Psi = |\psi\rangle\langle\psi|$.

For a pure state $|\psi\rangle = \sum_{i=1}^d \sqrt{\lambda_i} |a_i b_i\rangle$ with Schmidt coefficients $\{\lambda_i\}$, the unified $(q,s)$-concurrence is given by
\begin{equation}
C_{q,s}(|\psi\rangle) =1 -(\tr(\rho_A^q))^s =1 - (\sum_{i=1}^d \lambda_i^q)^s,
\label{eq:Aqs_concurrence_pure}
\end{equation}
where $q\geq 1\ \text{and } qs\geq 1$.

To evaluate $f_{\Psi^+}(|\psi\rangle)$, we express $|\psi\rangle$ in its Schmidt decomposition as
$$|\psi\rangle = \sum_{i=1}^d \sqrt{\lambda_i} |a_i b_i\rangle = (U_A \otimes U_B) \sum_{i=1}^d \sqrt{\lambda_i} |i i\rangle.$$
A straightforward calculation yields~\cite{2000Entanglement}
$$f_{\Psi^+}(|\psi\rangle) = \frac{1}{d} \left| \sum_{i=1}^d \sqrt{\lambda_i} v_{ii} \right|^2,$$
where $V = U_A^\mathrm{T} U_B$ and $v_{jk} = \langle j| V |k\rangle$.

It is straightforward to obtain $\xi(F,q,s,d)$ for $F \in (0, 1/d]$ by choosing $\lambda_1 = 1$ and $v_{11} =\sqrt{F}$, which yields $\xi(F,q,s, d) = 0$. In the entangled regime $F\geq 1/d$ (equivalently $Fd \geq 1$), we seek the minimum of Eq.~(\ref{eq:Aqs_concurrence_pure}) subject to the constraints
\begin{eqnarray}
  \sum_i \lambda_i &=& 1, \label{eq:constraint1} \\
\sum_i \sqrt{\lambda_i} &=& \sqrt{Fd}, \label{eq:constraint2}
\end{eqnarray}
This constrained optimization problem is addressed via the method of Lagrange multipliers \cite{rungta2003concurrence}.
The condition for an extremum is then
given by
\begin{equation}
(\sqrt{\lambda_i})^{2q - 1} + \mu_1 \sqrt{\lambda_i} + \mu_2 = 0, \label{eq:extremum_condition}
\end{equation}
where $\mu_1$ and $\mu_2$ are Lagrange multipliers. Note that for $q \geq 1$, the function $f(\sqrt{\lambda_i}) = (\sqrt{\lambda_i})^{2q - 1}$ is convex with respect to $\sqrt{\lambda_i}$. Since a convex function and a linear function can intersect in at most two points, Eq.~(\ref{eq:extremum_condition}) admits at most two distinct positive solutions for $\sqrt{\lambda_i}$, which we denote these solutions by $\gamma$ and $\delta$. The Schmidt vector $\vec{\lambda} = (\lambda_1, \lambda_2, \dots, \lambda_d)$ then takes the following form
\begin{equation}
\lambda_j =
\begin{cases}
\gamma^2, & j = 1, \cdots, n, \\
\delta^2, & j = n+1, \cdots, n+m, \\
0, & j = n+m+1, \cdots, d,
\end{cases}
\label{eq:schmidt_vector_form}
\end{equation}
where $n + m \leqslant d$ and $n \geqslant 1$. The minimization problem thus reduces to the following:

For given integers $n, m$ such that $n + m \leqslant d$,
\begin{eqnarray}
\xi(F,q,s,d)&=&\min\ C_{q,s}(|\psi\rangle)
\nonumber\\
&=& \min  (1 - (n\gamma^{2q} + m\delta^{2q})^s), \label{eq:min_problem} \\
\text{subject to}&& n\gamma^2 + m\delta^2 = 1,\nonumber  \\
\qquad\qquad && n\gamma + m\delta = \sqrt{F d}. \label{eq:constraint_nm2}
\end{eqnarray}

Solving Eq.~(\ref{eq:constraint_nm2}), we obtain the following solutions
\begin{equation}
\gamma_{nm}^\pm(F) = \frac{n\sqrt{F d} \pm \sqrt{nm(n + m - F d)}}{n(n + m)}, \label{eq:gamma_solution}
\end{equation}
and
\begin{equation}
\delta_{nm}^\pm(F) = \frac{\sqrt{F d} - n\gamma_{nm}^\pm}{m} = \frac{m\sqrt{F d} \mp \sqrt{nm(n + m - F d)}}{m(n + m)}. \label{eq:delta_solution}
\end{equation}

Since $\gamma_{nm}^- = \delta_{nm}^+$, the expression in Eq.~(\ref{eq:min_problem}) takes the same value for both $\gamma_{nm}^+$ and $\gamma_{nm}^-$. We can therefore restrict our attention to the solution $\gamma_{nm} := \gamma_{nm}^+$. For $\gamma_{nm}$ to be a valid solution, the expression under the square root in Eq.~(\ref{eq:gamma_solution}) must be non-negative, which implies $Fd \leq n + m$. Additionally, the non-negativity of $\delta_{nm}$ in Eq.~(\ref{eq:delta_solution}) requires $Fd \geq n$. Within this region, one can verify that $\delta_{nm}(F) \leq \sqrt{F d}/(n + m) \leq \gamma_{nm}(F)$. Note that $n = 0$ is not allowed; hence, we must have $n \geq 1$.

To find the minimum of $C_{q,s}(|\psi\rangle)$ over all valid choices of $n$ and $m$, we treat $n$ and $m$ as continuous variables and minimize over the parallelogram defined by $1 \leq n \leq F d$ and $F d \leq n + m \leq d$. This region reduces to a line when $F d = 1$, corresponding to the separability boundary. Within the parallelogram, we have $\gamma_{nm} \geq \delta_{nm} \geq 0$, with equality $\gamma_{nm} = \delta_{nm}$ occurring if and only if $n + m = F d$, and $\delta_{nm} = 0$ if and only if $n = F d$.

By analyzing the derivatives of $\gamma_{nm}$ and $\delta_{nm}$ with respect to $n$ and $m$ via differentiation of the constraints in Eq. (\ref{eq:constraint_nm2}), we obtain that
\begin{eqnarray}
  \frac{\partial \gamma}{\partial n} &=& \frac{1}{2n} \frac{2\gamma \delta - \gamma^2}{\gamma - \delta}, \\
\frac{\partial \delta}{\partial n} &=& -\frac{1}{2m} \frac{\gamma^2}{\gamma - \delta}, \\
\frac{\partial \delta}{\partial m} &=& -\frac{1}{2m} \frac{2\gamma \delta - \gamma^2}{\gamma - \delta}, \\
\frac{\partial \gamma}{\partial m} &=& \frac{1}{2n} \frac{\delta^2}{\gamma - \delta}.
\label{eq:partial_derivatives}
\end{eqnarray}

These can be used in Eq.(\ref{eq:min_problem}) to calculate the  derivatives of $C_{q,s}(\psi)$ with respect to $n$ and $m$ as
    \begin{eqnarray}
  \frac{\partial C_{q,s}}{\partial n} &=&-s(n\gamma^{2q}+m\delta^{2q})^{s-1}\Big(\gamma^{2q}+\frac{q\gamma^{2q-1}(2\delta\gamma-\gamma^2)-q\delta^{2q-1}\gamma^2}{\gamma-\delta}\Big),
\end{eqnarray}
and
\begin{eqnarray}
 \frac{\partial C_{q,s}}{\partial m}&=& -s(n\gamma^{2q}+m\delta^{2q})^{s-1}\left(\delta^{2q}+\frac{q\gamma^{2q-1}\delta^{2}-q\delta^{2q-1}(2\gamma\delta-\gamma^2)}{\gamma-\delta}\right).
\label{eq:Cqspartial_derivatives}
\end{eqnarray}

Let $u = m-n$ and $v = m + n$, where $u$ and $v$ represent the motions parallel and perpendicular to the boundaries of the parallelogram defined by
$m + n = c$ (with $c$ being a constant). The derivative of $C_{q,s}(\ket{\psi})$ with respect to $u$ is expressed as
 \begin{eqnarray}
    \frac{\partial C_{q,s}}{\partial u}&=&\frac{\partial C_{q,s}}{\partial n}\cdot\frac{\partial n}{\partial u}+\frac{\partial C_{q,s}}{\partial m}\cdot\frac{\partial m}{\partial u}\nonumber\\
    &=&\frac{1}{2}s(n\gamma^{2q}+m\delta^{2q})^{s-1}\Big(\gamma^{2q}+\frac{q\gamma^{2q-1}(2\delta\gamma-\gamma^2)-q\delta^{2q-1}\gamma^2}{\gamma-\delta}-\delta^{2q}-\frac{q\gamma^{2q-1}\delta^{2}-q\delta^{2q-1}(2\gamma\delta-\gamma^2)}{\gamma-\delta}\Big)\nonumber\\
    &=&\frac{1}{2}s(n\gamma^{2q}+m\delta^{2q})^{s-1}((1-q)\gamma^{2q}-\delta^{2q}+q\gamma^{2q-1}\delta-2q\gamma\delta^{2q-1}).
\label{eq:partialcqsm}
\end{eqnarray}

 Under the condition  $q\geq 1$ and $qs\geq 1$, we analyze $\frac{\partial C_{q,s}}{\partial m}$ in Eq.\eqref{eq:Cqspartial_derivatives}. Then, it follows that
\begin{eqnarray}
     \frac{\partial C_{q,s}}{\partial m}
     &\overset{(a)}{\leq}&-sq(n\gamma^{2q}+m\delta^{2q})^{s-1}\left(\frac{\gamma^{2q-1}\delta^{2}-2\delta^{2q}\gamma+\delta^{2q-1}\gamma^2}{\gamma-\delta}\right)\nonumber\\
&\overset{(b)}{\leq}&-sq(n\gamma^{2q}+m\delta^{2q})^{s-1}\left(\frac{\delta^{2q+1}-2\delta^{2q}\gamma+\delta^{2q-1}\gamma^2}{\gamma-\delta}\right)\nonumber\\
&=&-sq(n\gamma^{2q}+m\delta^{2q})^{s-1}\delta^{2q-1}({\gamma-\delta})\nonumber\\
&\overset{(c)}{\leq}&0
\label{eq:Cqspartial_derivatives1}
\end{eqnarray}
where the inequality in (a) is validated due to the fact that $\delta^{2q}\geq0$;
(b) holds for $\gamma\geq \delta$; (c) holds from $(n\gamma^{2q}+m\delta^{2q})^{s-1}>0$ for $q>1$ and $qs\geq 1$.

Next, we consider  $\frac{\partial C_{q,s}}{\partial u}$ in Eq.\eqref{eq:partialcqsm} and let $f(x)=(1-q)x^{2q}+qx^{2q-1}-2qx-1$ with $x\geq 1$ and $q>1$. We obtain that
\begin{equation}
    \frac{\partial f(x)}{\partial x}=qg(x),
\end{equation}
where $g(x)=2(1-q)x^{2q-1}+(2q-1)x^{2q-2}-2$. Therefore, we have that $g(x)$ is a decreasing function, i.e.,$\frac{\partial g(x)}{\partial x}=2(1-q)(2q-1)(x-1)x^{2q-3}\leq 0$ for $x\geq 1$ and $q>1$. Its maximum occurs at $x = 1$, yielding $g(1)=-1 < 0$. Hence, $\frac{\partial f(x)}{\partial x} < 0$, so $f(x)$ is decreasing for $x \geq 1$ and $q >1$. Since $\gamma \geq \delta \geq 0$, we have $\frac{\gamma}{\delta}\geq1$, and thus
\begin{eqnarray}
    f(\frac{\gamma}{\delta})&=&(1-q)\left(\frac{\gamma}{\delta}\right)^{2q}+q\left(\frac{\gamma}{\delta}\right)^{2q-1}-2q\left(\frac{\gamma}{\delta}\right)-1\nonumber\\
    &\leq& f(1)
    =-2q<0.
    \label{eq:fx}
\end{eqnarray}
Therefore, by inequality (\ref{eq:fx}), we can obtain that
   \begin{eqnarray}
    \frac{\partial C_{q,s}}{\partial u}&=&\frac{1}{2}s(n\gamma^{2q}+m\delta^{2q})^{s-1}
    ((1-q)\gamma^{2q}-\delta^{2q}+q\gamma^{2q-1}\delta-2q\gamma\delta^{2q-1})\nonumber\\
    &\leq&0.
\label{eq:cqsleq0}
\end{eqnarray}

From Eqs.(\ref{eq:Cqspartial_derivatives1}) and (\ref{eq:cqsleq0}), it is clear that within the parallelogram region, $\frac{\partial C_{q,s}}{\partial m} \leq 0$, and $\frac{\partial C_{q,s}}{\partial u} \leq 0$ except on the boundary \(m + n = Fd\), where the derivative vanishes. These monotonicity properties imply that the minimum of \(C_{q,s}(|\psi\rangle)\) is attained at the vertex where \(n = 1\) and \(m = d - 1\). Therefore, the minimal value of \(C_{q,s}(|\psi\rangle)\) is given by
\begin{equation}
\label{eq:xids}
    \xi(F,q,s,d) = 1 - \big(\gamma_{1,d-1}^{2q}+ (d - 1)\delta_{1,d-1}^{2q}\big)^s,
\end{equation}
where \(\gamma\) and \(\delta\) are defined by
\begin{eqnarray*}
  \gamma &=& \frac{1}{\sqrt{d}}\left[\sqrt{F} + \sqrt{(d - 1)(1 - F)}\right], \\
\delta &=& \frac{1}{\sqrt{d}}\left(\sqrt{F} - \frac{\sqrt{1 - F}}{\sqrt{d - 1}}\right).
\end{eqnarray*}
%%%%%%%%%%%%%%%%%%%%%%%%%%%%%%%%%%%%%%%%%%%%%%%%%%%%%%%%%%%%%%%%%%%%%%%%%%

\section{Analytical Expression of $\xi(\rho_w)$ for
Werner States} \label{sec:Appendix B}
Define the $(U \otimes U)$-twirling transformation as $T_{\text{wer}}(\rho) = \int \mathrm{d}U \, (U \otimes U) \, \rho \, (U^\dagger \otimes U^\dagger)$, which maps any density matrix $\rho$ to a Werner state \(\rho_w\) satisfying
\(T_{\text{wer}}(\rho) = \rho_w(\rho)\), where the Werner parameter is given by $$w(\rho) = \operatorname{tr} \left( \rho \sum_{i<k} |\Phi_{ik}^- \rangle \langle \Phi_{ik}^-| \right).$$
Moreover, the \(\rho_w\) remains invariant under this transformation, i.e., \(T_{\text{wer}}(\rho_w) = \rho_w\) \cite{1999Reduction,lee2003convex}. Applying the twirling operation \(T_{\text{wer}}\) to the pure state $|\psi\rangle = \sum_{i=1}^d \sqrt{\lambda_i} |a_i b_i\rangle$, where \(|\psi\rangle = \sum_{i=1}^d \sqrt{\lambda_i} U_A \otimes U_B |ii\rangle\), we obtain that
\begin{equation}
 T_{\text{wer}} (|\psi\rangle \langle \psi|) = \rho_{w (|\psi\rangle \langle \psi|)} = \rho_{w (\vec{\lambda}, \Lambda)},
\end{equation}
where \(\Lambda = U_A^\dagger U_B\) and the Werner parameter \(w(\vec{\lambda}, \Lambda)\) is given by
\begin{equation}
  w(\vec{\lambda}, \Lambda) = \sum_{i<k} |\langle \Psi_k^- |\psi\rangle|^2 = \frac{1}{2} \sum_{i<k} |\sqrt{\lambda_i} \Lambda_{ki} - \sqrt{\lambda_k} \Lambda_{ik}|^2,
\end{equation}
where \(\Lambda_{ik} = \langle i|\Lambda|k \rangle\). Then the function \(\xi\) can be expressed as
\begin{equation}
\label{eq:minwerner}
 \xi (\rho_w) = \min_{\{\vec{\lambda}, \Lambda\}} \left\{ C_{q,s} (\{\vec{\lambda}) : w(\vec{\lambda}, \Lambda) = w \right\}.
\end{equation}

Using \(w(\vec{\lambda}, \Lambda) = w\), we derive
\begin{eqnarray}
  2w &=& 1 - \sum_{i=1}^d \lambda_i |\Lambda_{ii}|^2 - 2 \sum_{i<k} \sqrt{\lambda_i \lambda_k} \operatorname{Re} (\Lambda_{ik} \Lambda_{ki}^*)\nonumber \\
&\leq& 1 + 2 \sum_{i<k} \sqrt{\lambda_i \lambda_k} |\operatorname{Re} (\Lambda_{ik} \Lambda_{ki}^*)| \nonumber\\
&\leq& 1 + 2 \sum_{i<k} \sqrt{\lambda_i \lambda_k} \nonumber\\
&=& \left|\sum_{i=1}^d \sqrt{\lambda_i}\right|^2,
\label{eq:Bwerner}
\end{eqnarray}
where \(\operatorname{Re}(z)\) denotes the real part of the complex number \(z\).

The equalities in Eq.~(\ref{eq:Bwerner}) hold only when the two nonzero components satisfy \(\Lambda_{01} = 1\) and \(\Lambda_{10} = -1\), and the vector \(\vec{\lambda} = (\lambda_1, \lambda_2, 0, \ldots, 0)\). Under these conditions, the minimization in Eq.~(\ref{eq:minwerner}) is achieved \cite{vollbrecht2001entanglement}, simplifying to
\begin{equation}
\label{eq:minsip}
\xi(\rho_w) = \min_{\vec{\lambda}} \left\{ C_{q,s} \left( \vec{\lambda} \right) : \left|\sum_{i=1}^2 \sqrt{\lambda_i}\right|^2 = 2w \right\}.
\end{equation}

For \( w \in (0, \frac{1}{2}] \), it is always possible to choose appropriate unitary transformations \( V_A \) and \( V_B \) such that \(\lambda_1 = 1\), leading to \(\xi(\rho_w) = 0\). For \( w > \frac{1}{2} \), the minimization in Eq.~(\ref{eq:minsip}) is performed under the constraints
\begin{eqnarray}
    \sum_{i=1}^2 \lambda_i &=& 1,  \\
\sum_{i=1}^2 \sqrt{\lambda_i} &=& \sqrt{2w}.
\end{eqnarray}

The subsequent calculation follows a procedure analogous to that in Appendix~\ref{sec:Appendix A}. By setting \( d = 2 \) and \( F = w \), we obtain the result
\begin{equation}
\xi(\varrho_w) =1 - \left[\left( \frac{1+G}{2}\right)^q +\left( \frac{1-G}{2} \right)^q \right]^s,
\label{eq:entanglement_measure}
\end{equation}
where $G = 2\sqrt{w(1-w)}$, $q\geq 1\ \text{and } qs\geq 1$.
%%%%%%%%%%%%%%%%%%%%%%%%%%%%%%%%%%%%%%%%%%%%%%%%%%%%%%%%%%%%%%%%%%%%%%%%%
% Bibliography
\bibliographystyle{unsrt}
\bibliography{main}
\fussy

\end{document}